\documentclass[10pt,journal,epsfig]{IEEEtran}

\usepackage[dvips]{graphicx}
\usepackage{graphicx}
\usepackage{amssymb}
\usepackage{cite}
\usepackage{subfigure}
\usepackage{amsmath}

\usepackage{algorithm}
\usepackage{algorithmic}

\begin{document}





\title{Phased Array-Based Sub-Nyquist Sampling for
Joint Wideband Spectrum Sensing and Direction-of-Arrival
Estimation}

\author{Feiyu Wang, Jun Fang, Huiping Duan, and Hongbin Li,~\IEEEmembership{Senior
Member,~IEEE}
\thanks{Feiyu Wang, and Jun Fang are with the National Key Laboratory
of Science and Technology on Communications, University of
Electronic Science and Technology of China, Chengdu 611731, China,
Email: JunFang@uestc.edu.cn}
\thanks{Huiping Duan is with the School of Electronic Engineering,
University of Electronic Science and Technology of China, Chengdu
611731, China, Email: huipingduan@uestc.edu.cn}
\thanks{Hongbin Li is
with the Department of Electrical and Computer Engineering,
Stevens Institute of Technology, Hoboken, NJ 07030, USA, E-mail:
Hongbin.Li@stevens.edu}
\thanks{This work was supported in part by the National Science
Foundation of China under Grant 61522104, and the National Science
Foundation under Grant ECCS-1408182 and Grant ECCS-1609393.}}

\maketitle


\begin{abstract}
In this paper, we study the problem of joint wideband spectrum
sensing and direction-of-arrival (DoA) estimation in a sub-Nyquist
sampling framework. Specifically, considering a scenario where a
few uncorrelated narrowband signals spread over a wide (say,
several GHz) frequency band, our objective is to estimate the
carrier frequencies and the DoAs associated with the narrowband
sources, as well as reconstruct the power spectra of these
narrowband signals. To overcome the sampling rate bottleneck for
wideband spectrum sensing, we propose a new phased-array based
sub-Nyquist sampling architecture with variable time delays, where
a uniform linear array (ULA) is employed and the received signal
at each antenna is delayed by a variable amount of time and then
sampled by a synchronized low-rate analog-digital converter (ADC).
Based on the collected sub-Nyquist samples, we calculate a set of
cross-correlation matrices with different time lags, and develop a
CANDECOMP/PARAFAC (CP) decomposition-based method for joint DoA,
carrier frequency and power spectrum recovery. Perfect recovery
conditions for the associated parameters and the power spectrum
are analyzed. Our analysis reveals that our proposed method does
not require to place any sparse constraint on the wideband
spectrum, only needs the sampling rate to be greater than the
bandwidth of the narrowband source signal with the largest
bandwidth among all sources. Simulation results show that our
proposed method can achieve an estimation accuracy close to the
associated Cram\'{e}r-Rao bounds (CRBs) using only a small number
of data samples.
\end{abstract}






\section{Introduction}
Wideband spectrum sensing, which aims to identify the frequency
locations of a few narrowband transmissions that spread over a
wide frequency band, has been of a growing interest in signal
processing and cognitive radio communications
\cite{AxellLeus12,SunNallanathan13}. To perform wideband spectrum
sensing, a conventional receiver requires to sample the received
signal at the Nyquist rate, which may be infeasible if the
spectrum under monitoring is very wide, say, reaches several GHz.
Also, a high sampling rate results in a large amount of data which
place a heavy burden on subsequent storage and processing. To
alleviate the sampling rate requirement, a variety of sub-Nyquist
sampling schemes, e.g.
\cite{MishaliEldar09,MishaliEldar10,MishaliEldar11,WakinBecker12},
were developed. The rationale behind such schemes is to exploit
the inherent sparsity in the frequency domain and formulate
wideband spectrum sensing as a sparse signal recovery problem
which, according to the compressed sensing theory
\cite{CandesRomberg06,Donoho06}, can perfectly recover the signal
of the entire frequency band based on compressed measurements or
sub-Nyquist samples. Furthermore, in
\cite{ArianandaLeus12,YenTsai13,CohenEldar14}, it was shown that
it is even possible to perfectly reconstruct the power spectrum
without placing any sparse constraint on the wideband spectrum
under monitoring.




In some applications such as electronic warfare, one need not only
conduct wideband spectrum sensing, but also identify the carrier
frequencies and directions-of-arrival (DoAs) associated with the
narrowband signals that live within the wide frequency band
\cite{ZoltowskiMathews94}. Besides, in massive MIMO or millimeter
wave systems where signals are transmitted via beamforming
techniques, the DoA information would allow a cognitive radio to
more efficiently exploit the vacant bands \cite{SteinYair15}. In
\cite{LemmaVeen98,LemmaVeen03}, ESPRIT-based methods were proposed
for joint carrier frequency and DoA estimation. These methods,
however, require the signal to be sampled at the Nyquist rate.
Recently, with the advent of compressed sensing theories, the
sparsity inherent in the spectral and spatial domains was utilized
to devise sub-Nyquist sampling-based algorithms for joint wideband
spectrum sensing and DoA estimation. Specifically, in
\cite{ArianandaLeus13}, a compressed sensing method was developed
in a phased array framework, where a multicoset sampling scheme is
executed at each antenna to collect non-uniform samples. In
practice, the multicoset sampling may be implemented using
multiple channels, with each channel delayed by a different time
offset and then sampled by a low-rate analog-digital converter
(ADC). Since the multicoset sampling has to be performed at each
antenna, the scheme \cite{ArianandaLeus13} involves a high
hardware complexity. In \cite{KumarRazul14,KumarRazul15}, a
simplified sub-Nyquist receiver architecture was proposed, in
which each antenna output is connected with only two channels,
i.e. a direct path and a delayed path. An ESPRIT-based algorithm
was then developed for joint DoA, carrier frequency, and signal
reconstruction. In addition to the above time delay-based
sub-Nyquist receiver architectures, an alternative sub-Nyquist
sampling approach, referred to as phased array-based modulated
wideband converter (MWC), was proposed in
\cite{SteinYair15,IoushuaYair17} for carrier and DoA estimation.
The receiver utilizes an L-shaped array, and all sensors have the
same sampling pattern implementing a single channel of the MWC.
Perfect recovery conditions were analyzed, and reconstruction
algorithms based on compressed sensing techniques were developed
in \cite{IoushuaYair17}.







In this paper, we propose a new sub-Nyquist receiver architecture,
referred to as the phased-array based sub-Nyquist sampling
architecture with variable time delays, for joint wideband
spectrum sensing and DoA estimation. Similar to
\cite{ArianandaLeus13,KumarRazul14,KumarRazul15}, the proposed
receiver architecture employs a uniform linear array. The received
signal at each antenna is delayed by a pre-specified time shift
and then sampled at a sub-Nyquist sampling rate. Compared with
existing sub-Nyquist receiver architectures, our proposed
sub-Nyquist scheme is simpler and easier to implement: it requires
only one ADC for each antenna output, thus leading to a lower
hardware complexity. Meanwhile, in our proposed architecture, the
time delays for different antennas can be arbitrary as long as
they satisfy a mild condition, which relaxes the requirement on
the accuracy of time delay lines. From the collected sub-Nyquist
samples, we calculate a set of cross-correlation matrices with
different time lags, based on which a third-order tensor that
admits a CANDECOMP/PARAFAC (CP) decomposition can be constructed.
We show that the DoAs and the carrier frequencies, along with the
power spectra associated with the sources, can be recovered from
the factor matrices. The perfect recovery condition is analyzed.
Our analysis shows that, to perfectly recover the power spectrum
of the wide frequency band and the associated parameters, we only
need the sampling rate to be greater than the bandwidth of the
narrowband source signal with the largest bandwidth among all
sources. In addition, our proposed method does not need to impose
any sparse constraint on the wideband spectrum. We also derive the
Cram\'{e}r-Rao bound (CRB) results for our estimation problem.
Simulation results show that our proposed method, with only a
small number of data samples, can achieve an estimation accuracy
close to the associated CRBs.




We notice that a CP decomposition-based approach was proposed in
\cite{SteinYair15} for joint DoA and carrier frequency estimation.
Different from our work, the construction of the tensor in
\cite{SteinYair15} has to rely on an L-shaped array and exploits
the cross-correlations between the two mutually perpendicular
sub-arrays. In addition, the PARAFAC analysis in
\cite{SteinYair15} can only help extract the DoA and carrier
frequency information, while in our proposed method, the DoA,
carrier frequency, and power spectrum associated with each source
can be simultaneously recovered from the CP decomposition.

The rest of the paper is organized as follows. In Section
\ref{sec:preliminaries}, we provide notations and basics on the CP
decomposition. The signal model and related assumptions are
discussed in Section \ref{sec:signal-model}. In Section
\ref{sec:architecture}, we propose a new phase-array based
sub-Nyquist receiver architecture. A CP decomposition-based method
for joint wideband spectrum sensing and DoA estimation is
developed in Section \ref{sec:proposed-method}. The uniqueness of
the CP decomposition is discussed in Section
\ref{sec:CP-uniqueness}, and the CRB analysis is conducted in
Section \ref{sec:CRB}. Simulation results are provided in Section
\ref{sec:experiments}, followed by concluding remarks in Section
\ref{sec:conclusion}.

\section{Preliminaries} \label{sec:preliminaries} To make the
paper self-contained, we provide a brief review on tensors and the
CP decomposition. More details regarding the notations and basics
on tensors can be found in \cite{KoldaBader09}. Simply speaking, a
tensor is a generalization of a matrix to higher-order dimensions,
also known as ways or modes. Vectors and matrices can be viewed as
special cases of tensors with one and two modes, respectively.
Throughout this paper, we use symbols $\otimes$ , $\circ$ , and
$\odot$ to denote the Kronecker, outer, and Khatri-Rao product,
respectively.

Let $\boldsymbol{\mathcal{X}}\in\mathbb{C}^{I_1\times
I_2\times\cdots\times I_N}$ denote an $N$th-order tensor with its
$(i_1,\ldots,i_N)$th entry denoted by $\mathcal{X}_{i_1\cdots
i_N}$. Here the order $N$ of a tensor is the number of dimensions.
Fibers are higher-order analogues of matrix rows and columns. The
mode-$n$ fibers of $\boldsymbol{\mathcal{X}}$ are
$I_n$-dimensional vectors obtained by fixing every index but
$i_n$. Slices are two-dimensional sections of a tensor, defined by
fixing all but two indices. Unfolding or matricization is an
operation that turns a tensor into a matrix. The mode-$n$
unfolding of a tensor $\boldsymbol{\mathcal{X}}$, denoted as
$\boldsymbol{X}_{(n)}$, arranges the mode-$n$ fibers to be the
columns of the resulting matrix. The CP decomposition decomposes a
tensor into a sum of rank-one component tensors (see Fig.
\ref{fig:CP}), i.e.
\begin{align}
\boldsymbol{\mathcal{X}}=
\sum\limits_{r=1}^{R}\lambda_r\boldsymbol{a}_r^{(1)}\circ\boldsymbol{a}_r^{(2)}\circ\cdots\circ\boldsymbol{a}_r^{(N)}
\end{align}
where $\boldsymbol{a}_r^{(n)}\in\mathbb{C}^{I_n}$, the minimum
achievable $R$ is referred to as the rank of the tensor, and
$\boldsymbol{A}^{(n)}\triangleq
[\boldsymbol{a}_{1}^{(n)}\phantom{0}\ldots\phantom{0}\boldsymbol{a}_{R}^{(n)}]\in\mathbb{C}^{I_n\times
R}$ denotes the factor matrix along the $n$-th mode. Elementwise,
we have
\begin{align}
\mathcal{X}_{i_1 i_2\cdots i_N}=\sum\limits_{r=1}^{R}\lambda_r
a_{i_1 r}^{(1)}a_{i_2 r}^{(2)}\cdots a_{i_N r}^{(N)}
\end{align}
The mode-$n$ unfolding of $\boldsymbol{\mathcal{X}}$ can be
expressed as
\begin{align}
\boldsymbol{X}_{(n)}=\boldsymbol{A}^{(n)}\boldsymbol{\Lambda}\left(\boldsymbol{A}^{(N)}
\odot\cdots\boldsymbol{A}^{(n+1)}\odot\boldsymbol{A}^{(n-1)}\odot\cdots\boldsymbol{A}^{(1)}\right)^T
\end{align}
where
$\boldsymbol{\Lambda}\triangleq\text{diag}(\lambda_1,\ldots,\lambda_R)$.

\begin{figure}[!t]
\centering
\includegraphics[width=9cm]{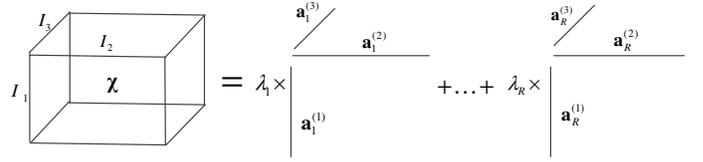}
\caption{Schematic of CP decomposition.} \label{fig:CP}
\end{figure}

\section{Signal Model} \label{sec:signal-model}
Consider a scenario in which $K$ uncorrelated, wide-sense
stationary, and far-field narrowband signals spreading over a wide
frequency band impinge on a wideband uniform linear array (ULA)
with $N$ receiver antennas, where we assume $N>K$. Let $s(t)$
denote the combination of the $K$ narrowband signals in the time
domain. $s(t)$ can be expressed as
\begin{align}
s(t)=\sum_{k=1}^K s_k(t) e^{j \omega_{k} t}
\end{align}
where $s_k(t)$ and $\omega_{k}\in\mathbb{R}^{+}$ denote the
complex baseband signal and the carrier frequency (in radians per
second) of the $k$th source signal, respectively. Each source
signal $s_k(t)$ is associated with an unknown azimuth DoA
$\theta_k\in[0,\pi)$. We have the following assumptions regarding
the source signals:
\begin{itemize}
\item[A1] The $K$ source signals $\{s_k(t)\}$ are assumed to be
mutually uncorrelated, wide-sense stationary, and bandlimited to
$[-B/2,B/2]$, i.e. $B_k\leq B,\forall k$, where $B_k$ denotes the
bandwidth of the $k$th source signal.
\item[A2] Sources either have distinct carrier frequencies $\{\omega_k\}$ or
distinct DoAs $\{\theta_k\}$, i.e. for any two source signals, we
have $(\theta_i,\omega_i)\neq (\theta_j,\omega_j), \forall i\neq
j$.
\item[A3] The multi-band signal $s(t)$
is bandlimited to $\mathcal{F}=[0,f_{\text{nyq}}]$, and we assume
$f_{\text{nyq}} \gg B$.
\end{itemize}

Assumption A2 is assumed to make signals distinguished from one
another. Note that this assumption is less restrictive than the
one made in other works, e.g. \cite{KumarRazul14,IoushuaYair17},
which, in order to remove the source ambiguity, require the
quantity $\omega_k\cos(\theta_k)$ to be mutually different for
different signals, i.e.
\begin{align}
\omega_i\cos(\theta_i)\neq \omega_j\cos(\theta_j) \quad \forall
i\neq j
\end{align}





After collecting the received signal at the array, our objective
is to jointly estimate the DoAs $\{\theta_k\}$, the carrier
frequencies $\{\omega_{k}\}$, as well as the power spectra
associated with the $K$ source signals. To accomplish this task,
we, in the following, propose a new phased-array based sub-Nyquist
receiver architecture. 



\begin{figure}[t]
   \centering
   \includegraphics [width=200pt]{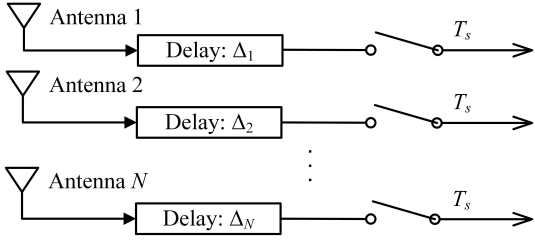}
   \caption{Proposed Phased-Array based Sub-nyquist Sampling Architecture with variable Time delays (PASSAT).}
   \label{fig:proposed-architecture}
\end{figure}

\section{Proposed Sub-Nyquist Receiver Architecture} \label{sec:architecture}
\subsection{Proposed Receiver Architecture}
In our receiver architecture, the received signal at each antenna
is delayed by a pre-specified factor $\Delta_n$ and then sampled
by a synchronous ADC with a sampling rate of $f_s=1/T_s$, where
$f_s\ll f_{\text{nyq}}$. We have the following assumptions
regarding the delay factors and the sampling rate:
\begin{itemize}
\item[A4] The time delay factors $\{\Delta_n\}$ can take arbitrary values as
long as the following condition holds valid
\begin{align}
(\Delta_{n+2}-2\Delta_{n+1}+\Delta_{n}) f_{\text{nyq}} < 1
\label{time-delay-assumption}
\end{align}
for some $n\in \{1,\ldots,N-2\}$.
\item[A5] The sampling rate $f_s$ is no less than the bandwidth of
the narrowband source signal which has the largest bandwidth among
all sources, i.e. $f_s\geq B$.
\end{itemize}

As will be shown later in our paper, Assumption A4 is essential to
identify the unknown carrier frequencies. Also, in practice, the
time delay factors $\{\Delta_n\}$ can be chosen to be of the same
order of magnitude as the Nyquist sampling interval such that the
narrowband approximation in (\ref{signal-model}) holds valid. The
proposed receiver architecture, termed as the Phased-Array based
Sub-Nyquist Sampling Architecture with variable Time delays
(PASSAT), is illustrated in Fig. \ref{fig:proposed-architecture}.
The analog signal observed by the $n$th antenna can be expressed
as
\begin{align}
x_n(t) = & \sum_{k=1}^K s_k(t-(n-1)\tau_k-\Delta_n) \nonumber \\
& \quad \ \times e^{j \omega_{k} (t-(n-1)\tau_k-\Delta_n)} + w_n(t) \nonumber \\
\approx & \sum_{k=1}^K s_k(t) e^{j \omega_{k}
(t-(n-1)\tau_k-\Delta_n)} + w_n(t) \label{signal-model}
\end{align}
where the approximation is due to the narrowband assumption,
$w_n(t)$ represents the additive white Gaussian noise with zero
mean and variance $\sigma^2$, and $\tau_k$ denotes the delay
between two adjacent sensors for a plane wave arriving in the
direction $\theta_k$ and is given by
\begin{align}
\tau_k=\frac{d\cos{\theta_k}}{C} \label{tau}
\end{align}
Here $d$ denotes the distance between two adjacent antennas and we
assume
\begin{itemize}
\item[A6] The distance between two adjacent antennas $d$ satisfies $d<C/f_{\text{nyq}}$,
where $C$ is the speed of light.
\end{itemize}

We will show later that this assumption is essential for the
recovery of the DoAs.

In practice, only the real part of $x_n(t)$ is observed and
sampled. Nevertheless, the corresponding imaginary part
$\Im[x_n(t)]$ can be retrieved from the real part $\Re[x_n(t)]$ by
passing the signal through a finite impulse response (FIR) Hilbert
transformer. The complex analytic signal can also be roughly
approximated by computing the discrete Fourier transform (DFT) of
the output of each antenna and throwing away the negative
frequency portion of the spectrum \cite{ZoltowskiMathews94}.


\begin{figure}[!t]
 \centering
\subfigure[]{\includegraphics[width=200pt]{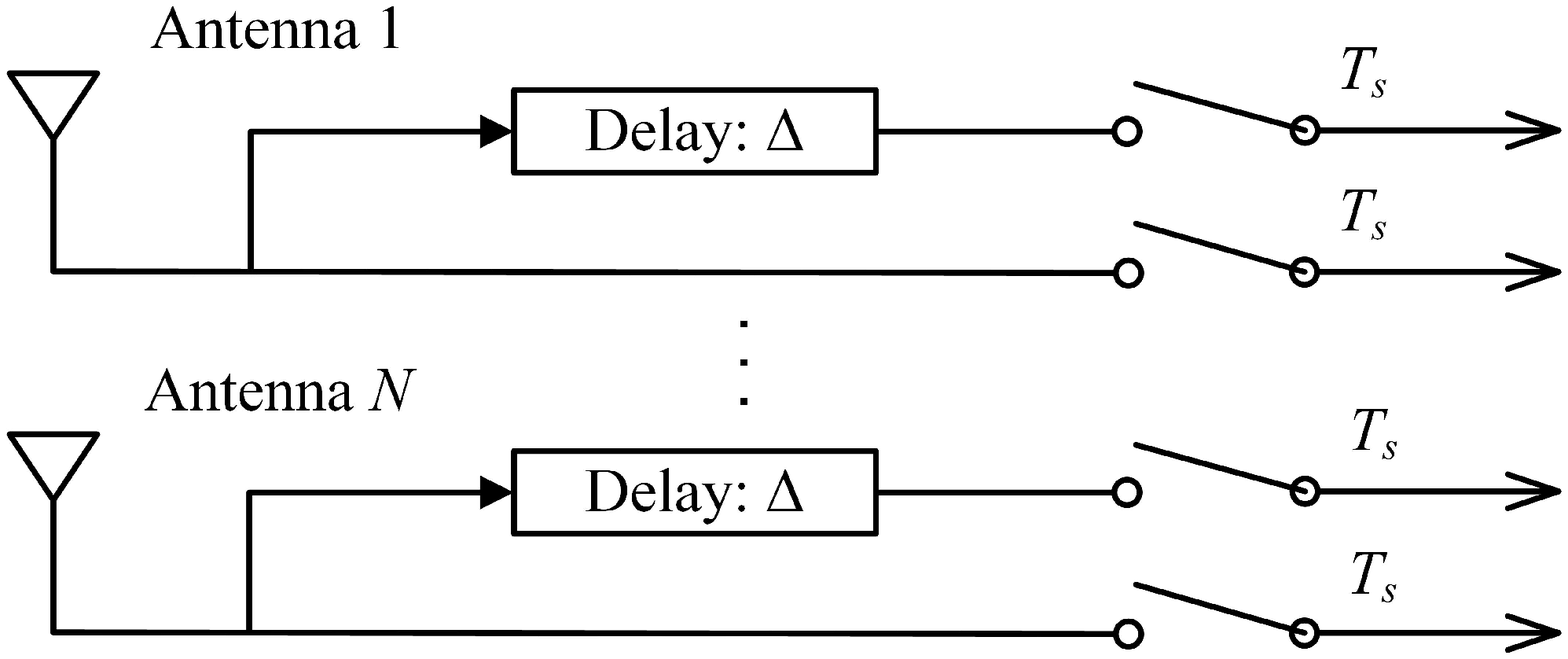}} \\
\subfigure[]{\includegraphics[width=200pt]{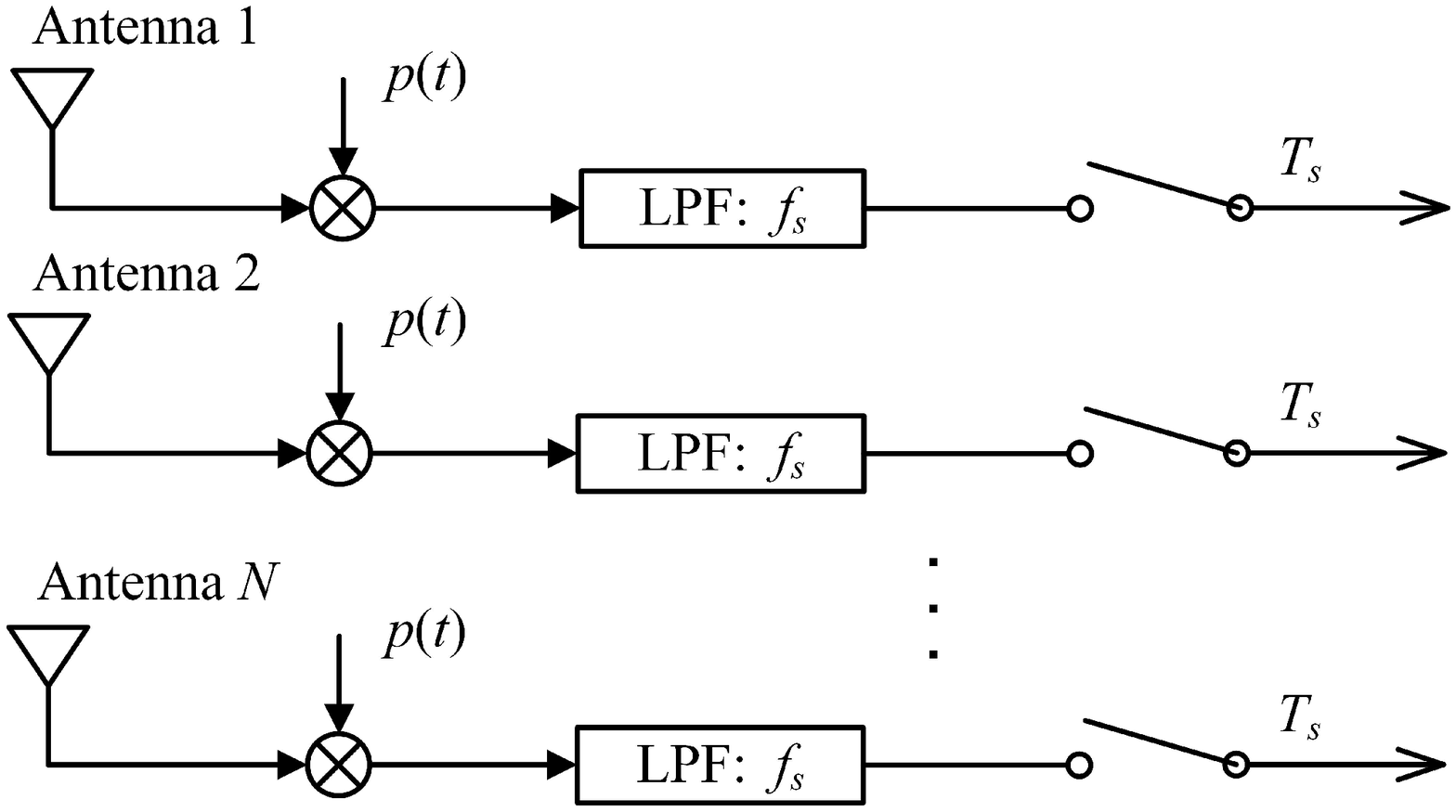}}
  \caption{Existing sub-Nyquist receiver architectures.
  (a) Sub-Nyquist sampling architecture proposed in \cite{KumarRazul14}.
  (b) Sub-Nyquist sampling architecture proposed in \cite{IoushuaYair17}.}
   \label{fig:existing-architecture}
\end{figure}

\subsection{Relation to and Distinction from Existing Architectures}
We notice that a time delay-based sub-Nyquist architecture was
also introduced in \cite{ArianandaLeus13,KumarRazul14,
KumarRazul15}. Nevertheless, there are two key distinctions
between our architecture and theirs. Firstly, our architecture has
a simpler structure with only $N$ delay channels, whereas the
architecture proposed in \cite{KumarRazul14} (see Fig.
\ref{fig:existing-architecture}(a)) requires $2N$ channels in
total, in which each antenna output passes through two channels,
namely, a direct path and a delayed path. As a consequence, the
number of required ADCs for the architecture \cite{KumarRazul14}
is twice the number of ADCs for our architecture. In
\cite{KumarRazul15}, a modified architecture was proposed based on
\cite{KumarRazul14}. It, however, still requires $2N$ channels,
with an $N$-channel delay network added to the first antenna.
Secondly, for our proposed architecture, the time delays can take
arbitrary values as long as the mild condition
(\ref{time-delay-assumption}) is satisfied. In contrast, for other
architectures, e.g. \cite{ArianandaLeus13,KumarRazul14,
KumarRazul15}, a precise time control is required such that the
time delays across different channels are strictly identical
\cite{KumarRazul14, KumarRazul15}, or the time delays must be
integer multiples of the Nyquist sampling interval
\cite{ArianandaLeus13}. Due to the inaccuracy caused by the time
shift elements, maintaining accurate time delays on the order of
the Nyquist sampling interval is difficult. The inaccuracy in
these delays will impair the recovery performance. Our
architecture is free from this issue because it allows more
flexible time delays and we can use the actual time delays
measured in practice for our proposed recovery algorithm.




In \cite{SteinYair15,IoushuaYair17}, a phased-array MWC-based
sub-Nyquist sampling architecture (see Fig.
\ref{fig:existing-architecture}(b)) was proposed for joint
wideband spectrum sensing and DoA estimation, in which an L-shaped
array composed of $2N+1$ sensors is adopted, and the output of
each sensor is multiplied by a same periodic pseudo-random
sequence, low-pass filtered and then sampled at a low rate.
Compared to the phased-array MWC-based sub-Nyquist sampling
architecture, our proposed delay-based scheme is much simpler to
implement. In \cite{LavrenkoRomer16}, it is argued that the
delay-based architectures suffer two major disadvantages which
include the need for high-precision delay lines as well as
specialized ADCs with high analog bandwidth. Nevertheless, as
discussed above, our proposed architecture, different from other
delay-based schemes \cite{ArianandaLeus13,KumarRazul14,
KumarRazul15}, has a relaxed requirement on the precision of delay
lines. Regarding the latter issue, it is known that there is an
inherent bandwidth limitation for practical ADCs, termed analog
(full-power) bandwidth, which determines the highest frequency
that can be handled by the device. In spite of that, we notice
that the inherent bandwidth of some inexpensive, low-end
commercial ADCs such as ADC12DC105 can reach up to 1GHz, while
some high-end ADCs with affordable prices, such as ADC12D500, have
an inherent bandwidth up to 2.7GHz, which may accommodate most
wideband spectrum sensing applications.

\section{Proposed CP Decomposition-Based Method} \label{sec:proposed-method}
Let $\delta(\cdot)$ denote the indicator function defined as
\begin{align}
\delta(x) = \left\{ \begin{array}{ll}
1, & \textrm{$x=0$}\\
0, & \textrm{$x\neq0$}
\end{array}. \right.
\end{align}
We first calculate the cross-correlation between two sensor
outputs $x_{m}(t_1)$ and $x_{n}(t_2)$. Recalling Assumption A1, we
have
\begin{align}
R_{m n}^{x} (t_1,t_2) &= \mathbb{E} \left[ x_{m}(t_1)x_{n}^{*}(t_2)    \right]  \nonumber \\
&=\sum_{k=1}^K R_{k}^{s} (t_1,t_2) a_{m k} a_{n k}^{*} +
R_{mn}^{w} (t_1,t_2) \label{cross-correlation}
\end{align}
where
\begin{align}
R_{k}^{s} (t_1,t_2)\triangleq\mathbb{E}\left[s_k(t_1) e^{j
\omega_{k} t_1} s_k^{*}(t_2) e^{-j \omega_{k} t_2}\right]
\label{source-correlation}
\end{align}
denotes the autocorrelation of the $k$-th modulated source signal,
\begin{align}
R_{mn}^{w} (t_1,t_2)\triangleq\mathbb{E}\left[w_m(t_1)w_n^{*}(t_2)
\right] = \sigma^2\delta(m-n)\delta(t_1-t_2)
\end{align}
represents the autocorrelation of the additive noise and
\begin{align}
a_{n k}\triangleq e^{-j((n-1) \tau_k \omega_k + \Delta_n\omega_k)}
\label{a-definition}
\end{align}

Since the source signals are wide-sense stationary, the
autocorrelation $R_{k}^{s} (t_1,t_2)$ depends only on the time
difference $t_1-t_2$. As a result, the cross-correlation of the
sensor outputs $R_{m n}^{x} (t_1,t_2)$ depends on the time
difference $t_1-t_2$ as well. Let $T_s$ denote the sampling
interval of the ADCs. The time difference has to be an integer
multiple of the sampling interval, i.e. $t_1-t_2=l T_s$ for
$l=-L,\ldots,L$. For notational convenience, we define
\begin{align}
r_{m,n}^{x}(l)\triangleq & R_{m n}^{x} (t+l T_s,t) \nonumber\\
r_k^{s}(l)\triangleq & R_{k}^{s} (t+l T_s,t) \nonumber\\
r_{m,n}^{w}(l) \triangleq & R_{mn}^{w} (t+l T_s,t) \nonumber
\end{align}
We can therefore express (\ref{cross-correlation}) as a
discrete-time form:
\begin{align}
r_{m,n}^{x}(l)=\sum_{k=1}^K r_k^{s}(l) a_{m k} a_{n k}^{*}
+r_{m,n}^{w}(l)
\end{align}
for $l=-L,\ldots,L$ and $m,n=1,\ldots,N$.

Our objective is to recover the DoAs $\{\theta_k\}$, the carrier
frequencies $\{\omega_{k}\}$, as well as the power spectra
associated with the $K$ source signals based on the second-order
statistics $\{r_{m,n}^{x}(l)\}$. For each time lag $l$, we can
construct a correlation matrix $\boldsymbol{R}^{x}(l)$ with its
$(m,n)$th entry given by $r_{m,n}^{x}(l)$. Also, it can be easily
verified that
\begin{align}
\boldsymbol{R}^{x}(l)=\sum_{k=1}^K
r_k^{s}(l)\boldsymbol{a}_{k}\boldsymbol{a}_{k}^H +
\boldsymbol{R}^{w}(l) \label{eqn1}
\end{align}
where $\boldsymbol{R}^{w}(l)$ denotes the cross-correlation matrix
of the additive noise with its $(m,n)$th entry given by
$r_{m,n}^{w}(l)$, and
\begin{align}
\boldsymbol{a}_{k}\triangleq [a_{1k}\phantom{0}a_{2
k}\phantom{0}\ldots\phantom{0}a_{Nk}]^T \label{ak-definition}
\end{align}
Since a set of cross-correlation matrices
$\{\boldsymbol{R}^{x}(l)\}_{l=-L}^L$ are available, we can
naturally express this set of correlation matrices by a
third-order tensor
$\boldsymbol{\mathcal{R}}^x\in\mathbb{C}^{(2L-1)\times N\times N}$
whose three modes respectively stand for the time lag $l$ and the
antenna indices, and its $(l,m,n)$-th entry given by
$r_{m,n}^{x}(l)$. Notice from (\ref{eqn1}) that each slice of the
tensor $\boldsymbol{\mathcal{R}}^x$, $\boldsymbol{R}^{x}(l)$, is a
weighted sum of a common set of rank-one outer products. The
tensor $\boldsymbol{\mathcal{R}}^x$ thus admits a CP decomposition
which decomposes a tensor into a sum of rank-one component
tensors, i.e.
\begin{align}
\boldsymbol{\mathcal{R}}^x = \sum_{k=1}^{K} \boldsymbol{r}_k \circ
\boldsymbol{a}_k \circ \boldsymbol{a}_k^{*} +
\boldsymbol{\mathcal{R}}^w
\end{align}
where $\circ$ denotes the outer product,
$\boldsymbol{\mathcal{R}}^w\in\mathbb{C}^{(2L-1)\times N\times N}$
with its $(l,m,n)$-th entry given by $r_{m,n}^{w}(l)$, and
\begin{align}
\boldsymbol{r}_k \triangleq
[r_{k}^{s}(-L)\phantom{0}\dots\phantom{0} r_{k}^{s}(L)]^T
\end{align}

Define
\begin{align}
\boldsymbol{R} & \triangleq \left[\boldsymbol{r}_1\phantom{0}\dots\phantom{0}\boldsymbol{r}_K\right] \\
\boldsymbol{A} & \triangleq
\left[\boldsymbol{a}_1\phantom{0}\dots\phantom{0}\boldsymbol{a}_K\right]
\label{A-definition}
\end{align}
The three matrices
$\{\boldsymbol{R},\boldsymbol{A},\boldsymbol{A}^{*}\}$ are
referred to as factor matrices associated with the noiseless
version of $\boldsymbol{\mathcal{R}}^x$. We see that the
information about the parameters $\{\theta_k,\omega_{k}\}$ as well
as the power spectra can be extracted from the factor matrices.
Motivated by this observation, we propose a two-stage method which
consists of a CP decomposition stage whose objective is to
estimate the factor matrices and a parameter estimation stage
whose objective is to jointly recover the DoAs, carrier
frequencies, and the power spectra of sources based on the
estimated factor matrices.



\subsection{CP Decomposition}
We first consider the scenario where the number of sources, $K$,
is known or estimated \emph{a priori}. Clearly, the CP
decomposition can be accomplished by solving the following
optimization problem
\begin{align}
\min_{\boldsymbol{\hat{R}},\boldsymbol{\hat{A}}} \quad & \left\|
\boldsymbol{\mathcal{R}}^x - \sum_{k=1}^{K} \boldsymbol{\hat{r}}_k
\circ \boldsymbol{\hat{a}}_k \circ \boldsymbol{\hat{a}}_k^{*}
\right\|_F^2
\end{align}
where
$\boldsymbol{\hat{R}}=[\boldsymbol{\hat{r}}_1\phantom{0}\dots\phantom{0}\boldsymbol{\hat{r}}_K]$,
$\boldsymbol{\hat{A}}=[\boldsymbol{\hat{a}}_1\phantom{0}\dots\phantom{0}\boldsymbol{\hat{a}}_K]$,
and $\| \cdot \|_F$ denotes the Frobenius norm. On the other hand,
note that the CP decomposition is unique under a mild condition.
Therefore we can use a new variable $\boldsymbol{\hat{b}}_k$ to
replace $\boldsymbol{\hat{a}}_k^{*}$, which leads to
\begin{align}
\min_{\boldsymbol{\hat{R}},\boldsymbol{\hat{A}},\boldsymbol{\hat{B}}}
\quad & \left\| \boldsymbol{\mathcal{R}}^x - \sum_{k=1}^{K}
\boldsymbol{\hat{r}}_k \circ \boldsymbol{\hat{a}}_k \circ
\boldsymbol{\hat{b}}_k \right\|_F^2
\end{align}
where
$\boldsymbol{\hat{B}}\triangleq[\boldsymbol{\hat{b}}_1\phantom{0}\dots\phantom{0}\boldsymbol{\hat{b}}_K]$.
The above optimization can be efficiently solved through an
alternating least squares (ALS) procedure which alternatively
updates one of the factor matrices to minimize the data fitting
error while keeping the other two factor matrices fixed:
\begin{align}
\widehat{\boldsymbol{R}}^{(t)} & = \arg\min_{\boldsymbol{R}}
\left\| (\boldsymbol{R}_{(1)}^x)^{T} -
(\widehat{\boldsymbol{B}}^{(t-1)}\odot\widehat{\boldsymbol{A}}^{(t-1)}) \boldsymbol{R}^T \right\|_F^2  \\
\widehat{\boldsymbol{A}}^{(t)} & = \arg\min_{\boldsymbol{A}}
\left\| (\boldsymbol{R}_{(2)}^x)^{T} -
(\widehat{\boldsymbol{B}}^{(t-1)}\odot\widehat{\boldsymbol{R}}^{(t)}) \boldsymbol{A}^T \right\|_F^2  \\
\widehat{\boldsymbol{B}}^{(t)} & = \arg\min_{\boldsymbol{B}}
\left\| (\boldsymbol{R}_{(3)}^x)^{T} -
(\widehat{\boldsymbol{A}}^{(t)}\odot\widehat{\boldsymbol{R}}^{(t)})
\boldsymbol{B}^T \right\|_F^2
\end{align}
where $\boldsymbol{R}^x_{(n)}$ denotes the mode-$n$ unfolding of
$\boldsymbol{\mathcal{R}}^x$.

If the knowledge of the number of sources, $K$, is unavailable,
more sophisticated CP decomposition techniques (e.g.
\cite{BazerqueMateos13,RaiWang14,ZhaoZhang15}) can be employed to
jointly estimate the model order and the factor matrices. The
basic idea is to use low rank-promoting priors or functions to
automatically determine the CP rank of the tensor. In
\cite{BazerqueMateos13}, when the CP rank, $K$, is unknown, the
following optimization was employed for CP decomposition
\begin{align}
\mathop {\min
}\limits_{\boldsymbol{\hat{R}},\boldsymbol{\hat{A}},\boldsymbol{\hat{B}}}\quad
&\left\| \boldsymbol{\mathcal{R}}^x -\boldsymbol{\mathcal X}
\right\|_F^2 + \mu
\left(\text{tr}({\boldsymbol{\hat{R}}}{{\boldsymbol{\hat{R}}}^H})
+ \text{tr}({\boldsymbol{\hat{A}}}{{\boldsymbol{\hat{A}}}^H}) +
\text{tr}({\boldsymbol{\hat{B}}}{{\boldsymbol{\hat{B}}}^H})\right) \nonumber\\
\text{s.t.}\quad &
\boldsymbol{\mathcal{X}}=\sum\limits_{k=1}^{\hat{K}}\boldsymbol{\hat{r}}_k\circ
\boldsymbol{\hat{a}}_k\circ\boldsymbol{\hat{b}}_k \label{opt-1}
\end{align}
where $\hat{K}\gg K$ denotes an overestimated CP rank, $\mu$ is a
regularization parameter to control the tradeoff between
low-rankness and the data fitting error, $\boldsymbol{\hat{R}}=
[\boldsymbol{\hat{r}}_1\phantom{0}\ldots\phantom{0}\boldsymbol{\hat{r}}_{\hat{K}}]$,
$\boldsymbol{\hat{A}}=
[\boldsymbol{\hat{a}}_1\phantom{0}\ldots\phantom{0}\boldsymbol{\hat{a}}_{\hat{K}}]$,
and $\boldsymbol{\hat{B}}=
[\boldsymbol{\hat{b}}_1\phantom{0}\ldots\phantom{0}\boldsymbol{\hat{b}}_{\hat{K}}]$.
The above optimization (\ref{opt-1}) can still be solved by an ALS
procedure as follows
\begin{align}
{{\boldsymbol{\hat{R}}}^{(t)}} &= \arg\min_{\boldsymbol{\hat{R}}}
\left\| {\left[ {\begin{array}{*{20}{c}}
(\boldsymbol{R}_{(1)}^x)^{T}\\
\boldsymbol{0}
\end{array}} \right] - \left[ {\begin{array}{*{20}{c}}
{{{\boldsymbol{\hat{B}}}^{(t-1)}} \odot {{\boldsymbol{\hat{A}}}^{(t-1)}}}\\
{\sqrt \mu  {\boldsymbol{I}}}
\end{array}} \right]{{\boldsymbol{\hat{R}}}^T}} \right\|_F^2 \nonumber\\
{{\boldsymbol{\hat{A}}}^{(t)}} &= \arg\min_{\boldsymbol{\hat{A}}}
\left\| {\left[ {\begin{array}{*{20}{c}}
(\boldsymbol{R}_{(2)}^x)^{T}\\
\boldsymbol{0}
\end{array}} \right] - \left[ {\begin{array}{*{20}{c}}
{{{\boldsymbol{\hat{B}}}^{(t-1)}} \odot {{\boldsymbol{\hat{R}}}^{(t)}}}\\
{\sqrt \mu  {\boldsymbol{I}}}
\end{array}} \right]{{\boldsymbol{\hat{A}}}^T}} \right\|_F^2 \nonumber\\
{{\boldsymbol{\hat{B}}}^{(t)}} &= \arg\min_{\boldsymbol{\hat{B}}}
\left\| {\left[ {\begin{array}{*{20}{c}}
(\boldsymbol{R}_{(3)}^x)^{T}\\
\boldsymbol{0}
\end{array}} \right] - \left[ {\begin{array}{*{20}{c}}
{{{\boldsymbol{\hat{A}}}^{(t)}} \odot {{\boldsymbol{\hat{R}}}^{(t)}}}\\
{\sqrt \mu  {\boldsymbol{I}}}
\end{array}} \right]{{\boldsymbol{\hat{B}}}^T}} \right\|_F^2 \nonumber
\end{align}
The true CP rank of the tensor, $K$, can be estimated by removing
those negligible rank-one tensor components after convergence.


\subsection{Joint DoA, Carrier Frequency and Power Spectrum Estimation}
We discuss how to jointly recover the DoAs, carrier frequencies,
and power spectra of sources based on the estimated factor
matrices. As shown in the next subsection, the CP decomposition is
unique up to scaling and permutation ambiguities under a mild
condition. More precisely, the estimated factor matrices and the
true factor matrices are related as
\begin{align}
\boldsymbol{\hat{R}}=&\boldsymbol{R}\boldsymbol{\Lambda}_1\boldsymbol{\Pi}+\boldsymbol{E}_1
\\
\boldsymbol{\hat{A}}=&\boldsymbol{A}\boldsymbol{\Lambda}_2\boldsymbol{\Pi}+\boldsymbol{E}_2
\\
\boldsymbol{\hat{B}}=&\boldsymbol{A}^{*}\boldsymbol{\Lambda}_3\boldsymbol{\Pi}+\boldsymbol{E}_3
\end{align}
where
$\{\boldsymbol{\Lambda}_1,\boldsymbol{\Lambda}_2,\boldsymbol{\Lambda}_3\}$
are unknown nonsingular diagonal matrices which satisfy
$\boldsymbol{\Lambda}_1\boldsymbol{\Lambda}_2\boldsymbol{\Lambda}_3=\boldsymbol{I}$;
$\boldsymbol{\Pi}$ is an unknown permutation matrix; and
$\boldsymbol{E}_1$, $\boldsymbol{E}_2$, and $\boldsymbol{E}_3$
denote the estimation errors associated with the three estimated
factor matrices, respectively. The permutation matrix
$\boldsymbol{\Pi}$ can be ignored as it is common to all three
factor matrices. Also, since we have prior knowledge that columns
of $\boldsymbol{A}/\sqrt{N}$ have unit norm, the amplitude
ambiguity can be estimated and removed, in which case we can write
\begin{align}
\boldsymbol{\hat{R}}=&\boldsymbol{R}\boldsymbol{\tilde{\Lambda}}_1+\boldsymbol{\tilde{E}}_1
\\
\boldsymbol{\hat{A}}=&\boldsymbol{A}\boldsymbol{\tilde{\Lambda}}_2+\boldsymbol{\tilde{E}}_2
\\
\boldsymbol{\hat{B}}=&\boldsymbol{A}^{*}\boldsymbol{\tilde{\Lambda}}_3+\boldsymbol{\tilde{E}}_3
\end{align}
where
$\boldsymbol{\tilde{\Lambda}}_1,\boldsymbol{\tilde{\Lambda}}_2,\boldsymbol{\tilde{\Lambda}}_3$
are unknown nonsingular diagonal matrices with their diagonal
elements lying on the unit circle.


Notice that the $k$th column of $\boldsymbol{A}$ is characterized
by the DoA and carrier frequency associated with the $k$th source.
We now discuss how to estimate $\{\omega_k\}$ and $\{\tau_k\}$
from the estimated factor matrix $\boldsymbol{\hat{A}}$. Note that
$\boldsymbol{\hat{B}}$ is also an estimate of $\boldsymbol{A}$.
Therefore either $\boldsymbol{\hat{A}}$ or $\boldsymbol{\hat{B}}$
can be used to estimate $\{\omega_k\}$ and $\{\tau_k\}$. Let
$\boldsymbol{\hat{a}}_k$ denote the $k$-th column of
$\boldsymbol{\hat{A}}$, and write
\begin{align}
\boldsymbol{\tilde{\Lambda}}_2=\text{diag}\{e^{-j\varphi_1},\dots,e^{-j\varphi_K}\}
\end{align}
where $\{\varphi_k\} \in [0,2\pi)$ are unknown parameters. To
simplify our exposition, we ignore the estimation errors
$\boldsymbol{\tilde{E}}_1$, $\boldsymbol{\tilde{E}}_2$, and
$\boldsymbol{\tilde{E}}_3$.


Write $z=r e^{j\varphi}$, and define
\begin{align}
\arg(z)\triangleq\text{mod}(\varphi,2\pi) \qquad \arg(z)\in
[0,2\pi)
\end{align}
where $\text{mod}(a,b)$ is a modulo operator which returns the
remainder of the Euclidean division of $a$ by $b$. Recalling
(\ref{a-definition}), we have
\begin{align}
\eta_{nk} & \triangleq \text{mod} \left(-\arg(\hat{a}_{nk}),2\pi\right) \nonumber \\
& =\text{mod} \left((n-1)\tau_k \omega_k +
\Delta_n\omega_k+\varphi_k,2\pi\right)
\end{align}
where $\hat{a}_{nk}$ denotes the $n$th entry of
$\boldsymbol{\hat{a}}_k$. Let
\begin{align}
\boldsymbol{\eta}_k\triangleq
\left[\eta_{1k}\phantom{0}\dots\phantom{0}\eta_{Nk}\right]^T
\nonumber
\end{align}
and let $\boldsymbol{D}_p$ denote a difference matrix defined as
\begin{displaymath}
\boldsymbol{D}_p \triangleq \left( \begin{array}{ccccc}
-1 & 1 & 0 & \ldots & 0\\
0 & -1 & 1 & \ldots & 0\\
\vdots & \vdots & \ddots & \ddots & \vdots\\
0 & 0 & \ldots & -1 & 1\\
\end{array} \right) \in \mathbb{R}^{(p-1)\times p}
\end{displaymath}
To recover $\omega_k$, we conduct a two-stage difference operation
as follows
\begin{align}
\boldsymbol{\beta}_{k}^{(1)}=& \text{mod} (\boldsymbol{D}_{N}
\boldsymbol{\eta}_k, 2\pi) \\
\boldsymbol{\beta}_{k}^{(2)}=& \text{mod}
(\boldsymbol{D}_{N-1}\boldsymbol{\beta}_{k}^{(1)}, 2\pi)
\end{align}
It can be easily verified that entries of
$\boldsymbol{\beta}_{k}^{(1)}$ and $\boldsymbol{\beta}_{k}^{(2)}$
are respectively given as
\begin{align}
\beta_{nk}^{(1)} & =\text{mod} \left(\tau_k \omega_k + (\Delta_{n+1}-\Delta_n)\omega_k,2\pi\right), \phantom{0}
n=1,\ldots,N-1 \label{eqn3}\\
\beta_{nk}^{(2)} & =\text{mod}
\left((\Delta_{n+2}-2\Delta_{n+1}+\Delta_n)\omega_k,2\pi\right),
\phantom{0} n=1,\ldots,N-2 \label{eqn2}
\end{align}

From (\ref{eqn2}), we can see that the information about the
carrier frequency $\omega_k$ is extracted after performing the
two-stage difference operation. By properly devising the time
delay factors $\{\Delta_n\}$, we can ensure that for some $n_0\in
\{1,\ldots,N\}$, the condition (\ref{time-delay-assumption}) holds
valid, i.e.
\begin{align}
(\Delta_{n_0+2}-2\Delta_{n_0+1}+\Delta_{n_0}) f_{\text{nyq}} < 1
\label{time-delay-condition}
\end{align}
The above condition implies
\begin{align}
(\Delta_{n_0+2}-2\Delta_{n_0+1}+\Delta_{n_0}) \omega_{\text{max}}
<2\pi
\end{align}
where
$\omega_{\text{max}}\triangleq\max\{\omega_1,\dots,\omega_K\}$.
Therefore $\omega_k$ can simply be estimated as
\begin{align}
\hat{\omega}_k = \frac {\beta_{n_0,k}^{(2)}}
{\Delta_{n_0+2}-2\Delta_{n_0+1}+\Delta_{n_0}}
\end{align}
In fact, for a careful selection of time delay factors
$\{\Delta_n\}$, the condition (\ref{time-delay-assumption}) (i.e.
(\ref{time-delay-condition})) may be satisfied for different
choices of $n$. As a result, we can obtain multiple estimates of
$\hat{\omega}_k$. To improve the estimation performance, a final
estimate of $\hat{\omega}_k$ can be chosen as the average of these
multiple estimates.

Under Assumption A6, that is, $d<C/f_{\text{nyq}}$, we have
$\tau_k\omega_{\text{max}}<2\pi$. Thus, substituting the estimated
$\hat{\omega}_k$ back into (\ref{eqn3}), $\tau_k$ can be obtained
as
\begin{align}
\hat{\tau}_k = \frac { \text{mod} \left( \beta_{nk}^{(1)}-
(\Delta_{n+1}-\Delta_n)\hat{\omega}_k,2\pi\right)}
{\hat{\omega}_k}
\end{align}
Note that for each $\beta_{nk}^{(1)}, n=1,\ldots,N-1$, we can
obtain an estimate of $\tau_k$. Therefore multiple estimates of
$\tau_k$ can be collected. Again, an average operation can be
conducted to yield a final estimate of $\tau_k$. Based on
$\hat{\tau}_k$, an estimate of the associated DoA $\theta_k$ can
be readily obtained from (\ref{tau}).




We now discuss how to recover the power spectra of the sources
$\{s_k(t)\}$. Let $\tilde{r}_{k}^{s}(\tau) \triangleq
R_{k}^{s}(t+\tau,t)$, where $\tau\in\mathbb{R}$ can be any real
value. The power spectrum of the $k$th source can thus be
expressed as the Fourier transform of $\tilde{r}_{k}^{s}(\tau)$,
i.e.
\begin{align}
\tilde{S}_{k}(\omega) = \int_{-\infty}^{+\infty}
\tilde{r}_{k}^{s}(\tau) e^{-j\omega\tau} \mathrm{d}\tau
\end{align}
Let $S_k(\omega)$ denote the discrete-time Fourier transform
(DTFT) of the autocorrelation sequence
$\{r_k^s(l)\}_{l=-\infty}^{+\infty}$, i.e.
\begin{align}
S_k(\omega)=\sum_{l=-\infty}^{\infty}r_k^s(l)e^{-j\omega l T_s}
\end{align}
According to the sampling theorem, $\tilde{S}_{k}(\omega)$ and
$S_{k}(\omega)$ are related as follows
\begin{align}
S_{k}(\omega) = \frac{1}{T_s} \sum_{n=-\infty}^{+\infty}
\tilde{S}_{k}\left(\omega+n\frac{2\pi}{T_s}\right)
\end{align}
Under Assumption A5, i.e. $f_s\geq B\geq B_k$, the power spectrum
$\tilde{S}_{k}(\omega)$ can be perfectly recovered by filtering
$S_{k}(\omega)$ with a bandpass filter, i.e.
\begin{align}
\tilde{S}_{k}(\omega) = \left\{ \begin{array}{ll}
T_s S_k(\omega), & \omega\in [\omega_k-\pi f_s,\omega_k+\pi f_s]\\
0, & \omega \notin [\omega_k-\pi f_s,\omega_k+\pi f_s]
\end{array} . \right.
\end{align}

Given the estimated factor matrix $\boldsymbol{\hat{R}}$, the DTFT
of the autocorrelation sequence $\{r_k^s(l)\}$ can be approximated
as
\begin{align}
\hat{S}_k(\omega)=\sum_{l=-L}^{L}\hat{r}_k^s(l)e^{-j\omega l T_s}
\end{align}
When $L$ is chosen to be sufficiently large, the estimation error
due to the time lag truncation is negligible. Also, although there
exists a phase ambiguity between the estimated autocorrelation
sequence $\boldsymbol{\hat{r}}_k$ and the true autocorrelation
sequence $\boldsymbol{r}_k$, this phase ambiguity can be removed
by noting that the power spectrum $S_k(\omega)$ is real and
non-negative. In addition, the power spectrum of each source is
automatically paired with its associated DoA and carrier frequency
due to the reason that both $\boldsymbol{\hat{R}}$ and
$\boldsymbol{\hat{A}}$ experience a common permutation operation.



\section{Uniqueness of CP Decomposition} \label{sec:CP-uniqueness}
We see that the uniqueness of the CP decomposition is crucial to
our proposed method. It is well known that the essential
uniqueness of CP decomposition can be guaranteed by Kruskal's
condition \cite{Kruskal77}. Let $k_{\boldsymbol{X}}$ denote the
k-rank of a matrix $\boldsymbol{X}$, which is defined as the
largest value of $k_{\boldsymbol{X}}$ such that every subset of
$k_{\boldsymbol{X}}$ columns of the matrix $\boldsymbol{X}$ is
linearly independent. We have the following theorem concerning the
uniqueness of CP decomposition.

\newtheorem{theorem}{Theorem}
\begin{theorem}
Let $(\boldsymbol{X},\boldsymbol{Y},\boldsymbol{Z})$ be a CP
solution which decomposes a third-order tensor
$\boldsymbol{\mathcal{X}}\in\mathbb{C}^{d_1\times d_2\times d_3}$
into $p$ rank-one arrays, where
$\boldsymbol{X}\in\mathbb{C}^{d_1\times p}$,
$\boldsymbol{Y}\in\mathbb{C}^{d_2\times p}$, and
$\boldsymbol{Z}\in\mathbb{C}^{d_3\times p}$. Suppose the following
Kruskal's condition
\begin{align}
k_{\boldsymbol{X}}+k_{\boldsymbol{Y}}+k_{\boldsymbol{Z}}\geq 2p+2
\label{Kruskals-condition}
\end{align}
holds and there is an alternative CP solution $(
\boldsymbol{\widehat{X}},\boldsymbol{\widehat{Y}},\boldsymbol{\widehat{Z}}
)$ which also decomposes $\boldsymbol{\mathcal{X}}$ into $p$
rank-one arrays. Then we have
$\boldsymbol{\widehat{X}}=\boldsymbol{X}\boldsymbol{\Pi}\boldsymbol{\Lambda}_x$,
$\boldsymbol{\widehat{Y}}=\boldsymbol{Y}\boldsymbol{\Pi}\boldsymbol{\Lambda}_y$,
and
$\boldsymbol{\widehat{Z}}=\boldsymbol{Z}\boldsymbol{\Pi}\boldsymbol{\Lambda}_z$,
where $\boldsymbol{\Pi}$ is a unique permutation matrix and
$\boldsymbol{\Lambda}_x$, $\boldsymbol{\Lambda}_y$, and
$\boldsymbol{\Lambda}_z$ are unique diagonal matrices such that
$\boldsymbol{\Lambda}_x\boldsymbol{\Lambda}_y\boldsymbol{\Lambda}_z=\boldsymbol{I}$.
\end{theorem}
\begin{proof}
A rigorous proof can be found in \cite{StegemanSidiropoulos07}.
\end{proof}

Note that Kruskal's condition cannot hold when $R=1$. However, in
that case the uniqueness has been proven by Harshman
\cite{Harshman72}. Kruskal's sufficient condition is also
necessary for $R=2$ and $R=3$, but not for $R>3$
\cite{StegemanSidiropoulos07}.

From the above theorem, we know that if
\begin{align}
k_{\boldsymbol{R}}+k_{\boldsymbol{A}}+k_{\boldsymbol{A}^{*}}\geq
2K+2 \label{Kruskals-condition-v2}
\end{align}
then the CP decomposition of $\boldsymbol{\mathcal{R}}^x$ is
essentially unique. Since $\boldsymbol{A}^{*}$ is the complex
conjugate of $\boldsymbol{A}$, we only need to examine the k-ranks
of $\boldsymbol{A}$ and $\boldsymbol{R}$.

Note that the $(n,k)$th entry of $\boldsymbol{A}$ is given by
\begin{align}
a_{n k} = e^{-j((n-1) \tau_k \omega_k + \Delta_{n}\omega_k)}
\end{align}
which is a function of the time delay factor $\Delta_{n}$. It is
not difficult to design a set of time delay factors $\{\Delta_n\}$
such that $k_{\boldsymbol{A}}=K$. For example, we divide $N$
antennas into two groups $S_1=\{1,\dots,K\}$ and
$S_2=\{K+1,\dots,N\}$. We set the delay factors in the first group
to be linearly proportional to $n-1$, i.e. $\Delta_n=(n-1)\nu$ for
$n\in S_1$, where $\nu\ge 0$ is a constant. In this case, the
first $K$ rows of $\boldsymbol{A}$ form a Vandermonde matrix:
\begin{align}
\boldsymbol{A}_{[1:K,:]}=\text{Vand}(\tau_1\omega_1+\nu\omega_1,\dots,\tau_K\omega_K+\nu\omega_K)
\end{align}
where $\text{Vand}(\phi_1,\dots,\phi_K)$ is defined as
\begin{displaymath}
\text{Vand}(\phi_1,\dots,\phi_K) \triangleq \left[
\begin{array}{cccc}
e^{-j(0 \phi_1)} & \ldots & e^{-j(0 \phi_K)}\\
e^{-j(1 \phi_1)} & \ldots & e^{-j(1 \phi_K)}\\
\vdots & \ddots & \vdots\\
e^{-j((K-1) \phi_1)} & \ldots & e^{-j((K-1) \phi_K)}\\
\end{array} \right]
\end{displaymath}
Thus $\boldsymbol{A}$ is full column rank with
$k_{\boldsymbol{A}}=K$ as long as $\{\omega_k\tau_k+\nu\omega_k\}$
are distinct from each other. If we set $\nu=0$, we only need
$\{\omega_k\tau_k\}$, i.e. $\{\omega_k\cos\theta_k\}$, are
distinct from each other. For the case where the quantities
$\{\omega_k\cos\theta_k\}$ for different source signals may be
identical, we can set $\nu\neq 0$, in which case we still have
$k_{\boldsymbol{A}}=K$ provided that the carrier frequencies
$\{\omega_k\}$ are mutually different. In other words, as long as
Assumption A2 is satisfied, we can always set an appropriate value
of $\nu$ to ensure $k_{\boldsymbol{A}}=K$. For other more general
choices of time delay factors $\{\Delta_n\}$, it can be
numerically checked that the k-rank of $\boldsymbol{A}$ still
equals to $K$ with a high probability, although a rigorous proof
is difficult.




Since we have $k_{\boldsymbol{A}}=K$, we only need
$k_{\boldsymbol{R}}\geq 2$ in order to satisfy Kruskal's
condition. This condition $k_{\boldsymbol{R}}\geq 2$ is met if
every two columns of $\boldsymbol{R}$ are linearly independent.
Note that the $k$th column of $\boldsymbol{R}$,
$\boldsymbol{r}_k$, is a truncated autocorrelation sequence of the
$k$th modulated signal $s_k(t)e^{j\omega_k t}$. Clearly, if the
baseband signals $\{s_k(t)\}$ have distinct power spectra, then
any two columns of $\boldsymbol{R}$ are linearly independent,
which implies $k_{\boldsymbol{R}}\geq 2$. In practice, since
source signals usually have different bandwidths, the diverse
power spectra condition can be easily satisfied. Even if the
baseband signals $\{s_k(t)\}$ have identical power spectra, the
autocorrelation sequences of any two modulated signals
$\{s_{k_1}(t)e^{j\omega_{k_1}t},s_{k_2}(t)e^{j\omega_{k_2}t}\}$
could still be linearly independent as long as their carrier
frequencies satisfy
\begin{align}
\mod\{|\omega_{k_1}-\omega_{k_2}|,2\pi f_s\}\neq 0
\label{omega-condition}
\end{align}
The above condition ensures that autocorrelation sequences
$\{\boldsymbol{r}_k\}$ of different modulated signals have
distinct exponential terms $\{e^{j\omega_{k}l T_s}\}$ (see
(\ref{source-correlation})). Due to the randomness of locations of
the carrier frequencies, the condition (\ref{omega-condition}) is
very likely to be satisfied in practice. As a result, we have
$k_{\boldsymbol{R}}\geq 2$.


\section{CRB Analysis} \label{sec:CRB}
In this section, we develop Cram\'{e}r-Rao bound (CRB) results for
the joint DoA, carrier frequency, and power spectra estimation
problem considered in this paper. As is well known, the CRB is a
lower bound on the variance of any unbiased estimator
\cite{Kay93}. It provides a benchmark for evaluating the
performance of our proposed method. In addition, the CRB results
illustrate the behavior of the resulting bounds, which helps
understand the effect of different system parameters, including
the noise power $\sigma^2$, the number of antennas $N$ and the
number of samples $N_s$, on the estimation performance.



\subsection{Signal Model}
Recall that the analog signal at each antenna is sampled with a
sampling rate $f_s=1/T_s$. The sampled signal at the $n$th antenna
can be written as (cf. (\ref{signal-model}))
\begin{align}
x_n(lT_s) &= \sum_{k=1}^K s_k(lT_s) e^{j \omega_{k}
(lT_s-(n-1)\tau_k-\Delta_n)} + w_n(lT_s) \nonumber \\
&= \sum_{k=1}^K a_{nk} s_k(lT_s) e^{j \omega_{k} (lT_s)} +
w_n(lT_s)
\end{align}
The above signal model can be rewritten in a vector-matrix form as
\begin{align}
\boldsymbol{x}_l=\boldsymbol{A}\boldsymbol{s}_l+\boldsymbol{w}_l,\quad
l=0,\dots,N_s-1
\end{align}
where $\boldsymbol{A}$ is defined in (\ref{A-definition}),
$\boldsymbol{x}_l \triangleq
\left[x_1(lT_s)\phantom{0}\dots\phantom{0}x_{N}(lT_s)\right]^T$,
$\boldsymbol{w}_l \triangleq
\left[w_1(lT_s)\phantom{0}\dots\phantom{0}w_{N}(lT_s)\right]^T$,
and
\begin{align}
\boldsymbol{s}_l &\triangleq
\left[s_1(lT_s)e^{j\omega_1(lT_s)}\phantom{0}\dots\phantom{0}s_K(lT_s)e^{j\omega_K(lT_s)}\right]^T
\nonumber
\end{align}
Suppose we collect a total number of $N_s$ ($l=0,\ldots,N_s-1$)
samples. The received signal can thus be expressed as
\begin{align}
\boldsymbol{X}=\boldsymbol{A}\boldsymbol{S}+\boldsymbol{W}
\end{align}
where
\begin{align}
\boldsymbol{X}&\triangleq\left[\boldsymbol{x}_0\phantom{0}\dots\phantom{0}\boldsymbol{x}_{N_s-1}\right] \nonumber\\
\boldsymbol{S}&\triangleq\left[\boldsymbol{s}_0\phantom{0}\dots\phantom{0}\boldsymbol{s}_{N_s-1}\right] \nonumber\\
\boldsymbol{W}&\triangleq\left[\boldsymbol{w}_0\phantom{0}\dots\phantom{0}\boldsymbol{w}_{N_s-1}\right].
\nonumber
\end{align}
Let $\boldsymbol{x}\triangleq\text{vec}(\boldsymbol{X}^T)$, where
$\text{vec}(\boldsymbol{Z})$ denotes a vectorization operation
which stacks the columns of $\boldsymbol{Z}$ into a single column
vector. We have
\begin{align}
\boldsymbol{x}=\boldsymbol{\bar{A}}\boldsymbol{s}+\boldsymbol{w}
\label{vector-matrix}
\end{align}
where
\begin{align}
\boldsymbol{x}\triangleq\mathrm{vec}(\boldsymbol{X}^T),&
\quad \boldsymbol{w}\triangleq\mathrm{vec}(\boldsymbol{W}^T)\nonumber\\
\boldsymbol{s}\triangleq\mathrm{vec}(\boldsymbol{S}^T),& \quad
\boldsymbol{\bar{A}}\triangleq \boldsymbol{A}\otimes
\boldsymbol{I}_{N_s} \label{A-bar}
\end{align}
in which $\boldsymbol{I}_{n}$ denotes an $n\times n$ identity
matrix. We assume that
$\boldsymbol{w}\sim\mathcal{CN}(\boldsymbol{0},\sigma^2\boldsymbol{I}_{N\cdot
N_s})$ and
$\boldsymbol{s}\sim\mathcal{CN}(\boldsymbol{0},\boldsymbol{R}_s)$
follow a circularly-symmetric complex Gaussian distribution, where
$\boldsymbol{R}_s$ denotes the source covariance matrix which
needs to be estimated along with other parameters. Note that in
our proposed algorithm, the additive noise $\boldsymbol{w}$ and
the source signal $\boldsymbol{s}$ are not restricted to be
circularly-symmetric complex Gaussian. Here we make such an
assumption in order to facilitate the CRB analysis.


Under the assumption that $\boldsymbol{w}$ and $\boldsymbol{s}$
are circularly-symmetric complex Gaussian random variables, we can
readily verify that $\boldsymbol{x}$ also follows a
circularly-symmetric complex Gaussian distribution, i.e.
$\boldsymbol{x}\sim\mathcal{CN}(\boldsymbol{0},\boldsymbol{R}_x)$,
where
\begin{align}
\boldsymbol{R}_x  &\triangleq
\mathbb{E}\left[\boldsymbol{x}\boldsymbol{x}^{H}\right]
=\mathbb{E}\left[\boldsymbol{\bar{A}}\boldsymbol{s}\boldsymbol{s}^{H}
\boldsymbol{\bar{A}}^{H}\right]
+\mathbb{E}\left[\boldsymbol{w}\boldsymbol{w}^{H}\right] \nonumber\\
&=\boldsymbol{\bar{A}} \boldsymbol{R}_s \boldsymbol{\bar{A}}^H +
\sigma^2\boldsymbol{I}_{N N_s}
\end{align}
From Assumption A1, we know that $\boldsymbol{R}_s$ is a block
diagonal matrix, i.e.
\begin{align}
\boldsymbol{R}_s=\mathrm{diag}(\boldsymbol{P}_1,\dots,\boldsymbol{P}_K).
\label{Rs}
\end{align}
where $\boldsymbol{P}_k\triangleq\mathbb{E}
\left[\tilde{\boldsymbol{s}}_k
\tilde{\boldsymbol{s}}_k^{H}\right]$ denotes the autocorrelation
matrix of the $k$th signal, and $\tilde{\boldsymbol{s}}_k$ is the
transpose of the $k$th row of $\boldsymbol{S}$, i.e.
\begin{align}
\tilde{\boldsymbol{s}}_k \triangleq
\left[s_k(0T_s)e^{j\omega_k(0T_s)} \dots
s_k((N_s-1)T_s)e^{j\omega_k((N_s-1)T_s)}\right]^T \nonumber
\end{align}
Also, in Assumption A1, each source is assumed to be wide-sense
stationary. Therefore the autocorrelation matrix
$\boldsymbol{P}_k$ is a Hermitian-Toeplitz matrix. Here Toeplitz
means that it has diagonal-constant entries, i.e. each descending
diagonal from left to right is constant. Let $p_0^{k}$ denote the
constant for elements located on the main diagonal, and $p_l^{k}$,
$l\ge1$, denote the constant for elements located on the $l$th
diagonal below the main diagonal of $\boldsymbol{P}_k$. Let
\begin{displaymath}
\boldsymbol{T}_l \triangleq \left( \begin{array}{cc}
\boldsymbol{0} & \boldsymbol{I}_{N_s-l} \\
\boldsymbol{0} & \boldsymbol{0} \\
\end{array} \right) \in \mathbb{R}^{N_s\times N_s}
\end{displaymath}
and
\begin{displaymath}
\boldsymbol{T}_{-l} \triangleq \left( \begin{array}{cc}
\boldsymbol{0} & \boldsymbol{0} \\
\boldsymbol{I}_{N_s-l} & \boldsymbol{0} \\
\end{array} \right) \in \mathbb{R}^{N_s\times N_s}
\end{displaymath}
The autocorrelation matrix $\boldsymbol{P}_k$ can thus be
expressed as
\begin{align}
\boldsymbol{P}_k= p_{0}^{k} \boldsymbol{I}_{N_s}  + \sum_{l=1}^{L}
\left[ p_l^{k} \boldsymbol{T}_{-l} + (p_l^{k})^{*}
\boldsymbol{T}_l \right] \label{Pk}
\end{align}
where $L$ is chosen to be sufficiently large to ensure $p_l^{k}=0$
for $l>L$. From (\ref{Pk}), we can see that $\boldsymbol{P}_k$ is
characterized by parameters
\begin{align}
\boldsymbol{p}_k\triangleq \left[ p_{0}^{k} \phantom{0}
\Re(p_{1}^{k}) \phantom{0}\dots\phantom{0} \Re(p_{L}^{k})
\phantom{0} \Im(p_{1}^{k}) \phantom{0}\dots\phantom{0}
\Im(p_{L}^{k}) \right]
\end{align}
As a result, $\boldsymbol{R}_s$ is characterized by parameters
\begin{align}
\boldsymbol{p} \triangleq \left[ \boldsymbol{p}_1
\phantom{0}\dots\phantom{0} \boldsymbol{p}_K \right]
\end{align}

On the other hand, notice that $\boldsymbol{\bar{A}}$ is a
parameterized matrix, with each column of $\boldsymbol{A}$
determined by the DoA and the carrier frequency of each source,
i.e. $\{\theta_k,\omega_k\}$. Unfortunately, the value ranges for
the DoA and the carrier frequency differ by orders of magnitude,
which may cause numerical instability in computing the CRB matrix.
To address this difficulty, we, instead, analyze the CRB for the
following two parameters $\{\xi_k,\psi_k\}$ defined as
\begin{align}
\xi_k\triangleq \omega_k\tau_k \qquad \psi_k\triangleq \omega_k/c
\end{align}
where $c$ is a parameter appropriate chosen (e.g. $c=10^9$) such
that values of $\xi_k$ and $\psi_k$ roughly have the same scale.
Also, we define
\begin{align}
\boldsymbol{\xi} &\triangleq \left[ \xi_1 \phantom{0}\dots\phantom{0} \xi_K\right] \nonumber\\
\boldsymbol{\psi} &\triangleq \left[ \psi_1
\phantom{0}\dots\phantom{0} \psi_K\right] \nonumber
\end{align}

We see that the complete set of parameters to be estimated include
\begin{align}
\boldsymbol{\alpha} \triangleq \left[ \boldsymbol{\xi} \phantom{0}
\boldsymbol{\psi} \phantom{0} \boldsymbol{p} \phantom{0} \sigma^2
\right]
\end{align}
Recall that $\boldsymbol{x}$ follows a complex Gaussian
distribution with zero mean and covariance matrix
$\boldsymbol{R}_x$. Therefore the log-likelihood function of
$\boldsymbol{\alpha}$ can be expressed as
\begin{align}
L(\boldsymbol{\alpha})\propto - \ln{|\boldsymbol{R}_x|} -
\boldsymbol{x}^{H} \boldsymbol{R}_x^{-1} \boldsymbol{x}
\end{align}


\subsection{Calculation of The CRB Matrix}
Since the random vector $\boldsymbol{x}$ follows a
circularly-symmetric complex Gaussian distribution, we can resort
to the Slepian-Bangs formula
\cite{StoicaNehorai90,StoicaLarsson01} to compute the Fisher
information matrix (FIM). According to the Slepian-Bangs formula,
the $(i,j)$th element of the FIM $\boldsymbol{\Omega}$ is
calculated as
\begin{align}
\Omega_{ij}=\mathrm{tr}\left(
\boldsymbol{R}_x^{-1}\frac{\partial\boldsymbol{R}_x}{\partial\alpha_i}
\boldsymbol{R}_x^{-1}\frac{\partial\boldsymbol{R}_x}{\partial\alpha_j}
\right)
\end{align}
where $\alpha_i$ and $\alpha_j$ denote the $i$th and the $j$th
entries of $\boldsymbol{\alpha}$, respectively.

By utilizing the structures of $\boldsymbol{\bar{A}}$ and
$\boldsymbol{R}_s$ (cf. (\ref{A-bar}) and (\ref{Rs})),
$\boldsymbol{R}_x$ can be expressed as
\begin{align}
\boldsymbol{R}_x = &[\boldsymbol{a}_1\otimes \boldsymbol{I}_{N_s}
\phantom{0}\dots\phantom{0} \boldsymbol{a}_K\otimes
\boldsymbol{I}_{N_s}]
\cdot \mathrm{diag}(\boldsymbol{P}_1,\dots,\boldsymbol{P}_K)  \nonumber\\
& \cdot [\boldsymbol{a}_1\otimes \boldsymbol{I}_{N_s}
\phantom{0}\dots\phantom{0} \boldsymbol{a}_K\otimes
\boldsymbol{I}_{N_s}]^H
+ \sigma^2\boldsymbol{I}_{N\cdot N_s} \nonumber\\
= &[\boldsymbol{a}_1\otimes \boldsymbol{P}_1 \phantom{0}\dots\phantom{0} \boldsymbol{a}_K\otimes \boldsymbol{P}_K] \nonumber\\
& \cdot [\boldsymbol{a}_1\otimes \boldsymbol{I}_{N_s}
\phantom{0}\dots\phantom{0} \boldsymbol{a}_K\otimes
\boldsymbol{I}_{N_s}]^H
+ \sigma^2\boldsymbol{I}_{N\cdot N_s} \nonumber\\
= & \sum_{k=1}^{K} (\boldsymbol{a}_k\boldsymbol{a}_k^{H})\otimes
\boldsymbol{P}_k + \sigma^2\boldsymbol{I}_{N\cdot N_s} \label{Rx}
\end{align}
where $\boldsymbol{a}_k$, defined in (\ref{ak-definition}), is the
$k$th column of $\boldsymbol{A}$.

We first compute the partial derivative of $\boldsymbol{R}_x$ with
respect to $\xi_k$ and $\psi_k$. From (\ref{a-definition}) and the
definition of $\{\xi_k,\psi_k\}$, we can write
\begin{align}
a_{n k} = e^{-j( (n-1) \xi_k + c\Delta_{n}\psi_k)}
\end{align}
Thus we have
\begin{align}
\frac{\partial\boldsymbol{a}_k}{\partial\xi_k} &=-j\cdot\mathrm{diag}(0,\dots,N-1)\cdot\boldsymbol{a}_k
\label{eqn4} \\
\frac{\partial\boldsymbol{a}_k}{\partial\psi_k} &=-j \cdot c \cdot
\mathrm{diag}(\Delta_1,\dots,\Delta_{N})\cdot\boldsymbol{a}_k
\label{eqn5}
\end{align}
Combining (\ref{Rx}) and (\ref{eqn4})--(\ref{eqn5}), we have
\begin{align}
\frac{\partial\boldsymbol{R}_x}{\partial\xi_k}
&=\left(\frac{\partial\boldsymbol{a}_k}{\partial\xi_k}\boldsymbol{a}_k^{H}
+ \boldsymbol{a}_k
\frac{\partial\boldsymbol{a}_k^H}{\partial\xi_k} \right)
\otimes \boldsymbol{P}_k \\
\frac{\partial\boldsymbol{R}_x}{\partial\psi_k}
&=\left(\frac{\partial\boldsymbol{a}_k}{\partial\psi_k}\boldsymbol{a}_k^{H}
+ \boldsymbol{a}_k
\frac{\partial\boldsymbol{a}_k^H}{\partial\psi_k} \right) \otimes
\boldsymbol{P}_k .
\end{align}
Similarly, we can obtain the partial derivatives with respect to
other parameters as follows
\begin{align}
\frac{\partial\boldsymbol{R}_x}{\partial(p_{0}^{k})} &=\left(
\boldsymbol{a}_k \boldsymbol{a}_k^{H}\right)
\otimes \left( \frac{\partial\boldsymbol{P}_k}{\partial (p_{0}^{k})} \right) \nonumber\\
&=\left( \boldsymbol{a}_k \boldsymbol{a}_k^{H}\right)
\otimes \boldsymbol{I}_{N_s} \\
\frac{\partial\boldsymbol{R}_x}{\partial(\Re(p_{l}^{k}))} &=\left(
\boldsymbol{a}_k \boldsymbol{a}_k^{H}\right)
\otimes \left( \frac{\partial\boldsymbol{P}_k}{\partial (\Re(p_{l}^{k}))} \right) \nonumber\\
&=\left( \boldsymbol{a}_k \boldsymbol{a}_k^{H}\right)
\otimes \left(\boldsymbol{T}_{-l}+\boldsymbol{T}_{l}\right) \\
\frac{\partial\boldsymbol{R}_x}{\partial(\Im(p_{l}^{k}))} &=\left(
\boldsymbol{a}_k \boldsymbol{a}_k^{H}\right)
\otimes \left( \frac{\partial\boldsymbol{P}_k}{\partial (\Im(p_{l}^{k}))} \right) \nonumber\\
&=\left( \boldsymbol{a}_k \boldsymbol{a}_k^{H}\right) \otimes
\left(j\boldsymbol{T}_{-l}-j\boldsymbol{T}_{l} \right)
\end{align}
and
\begin{align}
\frac{\partial\boldsymbol{R}_x}{\partial(\sigma^2)}
=\boldsymbol{I}_{N\cdot N_s}.
\end{align}
After obtaining the FIM $\boldsymbol{\Omega}$, the CRB can be
calculated as \cite{Kay93}
\begin{align}
\mathrm{CRB}(\boldsymbol{\alpha})=\boldsymbol{\Omega}^{-1}.
\end{align}

\begin{figure}[!t]
 \centering
\subfigure[True and estimated carrier frequencies and DoAs.]
{\includegraphics[width=3.5in]{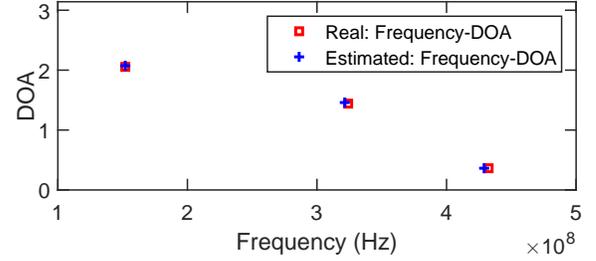}} \\
\subfigure[Original power spectra of sources.]
{\includegraphics[width=3.5in]{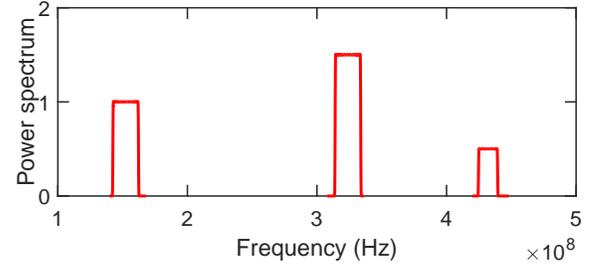}} \\
\subfigure[Estimated power spectra of sources.]
{\includegraphics[width=3.5in]{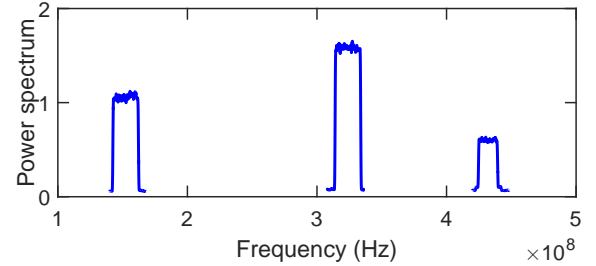}}
  \caption{True and estimated carrier frequencies, DoAs, and power spectra of sources, $\text{SNR}=5$dB.}
   \label{fig:example1}
\end{figure}

\section{Simulation Results} \label{sec:experiments}
In this section, we carry out experiments to illustrate the
performance of our proposed method. In our simulations, we set
$f_{\text{nyq}}=1\text{GHz}$. The distance between two adjacent
antennas, $d$, is set equal to $d=0.8\times C/f_{\text{nyq}}$ in
order to meet the condition in Assumption A6. The number of
antennas is set to $N=8$, and for simplicity, the time delay
factors are set as
\begin{align}
\Delta_n = \left\{
\begin{array}{ll}
0 \ \text{s}, & n=1,\dots,N/2\\
10^{-9} \ \text{s}, & n=N/2+1,\dots,N
\end{array} \right.
\end{align}
With this setup, the condition (\ref{time-delay-assumption}) can
be satisfied for $n=N/2-1$. The signal-to-noise ratio (SNR) is
defined as
\begin{align}
\text{SNR} \triangleq \frac{\mathbb{E}[|s(t)|^2]}{\sigma^2}
\end{align}

We first consider the case in which $K=3$ uncorrelated, wide-sense
stationary sources spreading over the wide frequency band $(0,
500]\text{MHz}$ impinge on a ULA of $N$ antennas. The DoAs of
these three sources are given respectively by $\theta_1=2.051$,
$\theta_2=1.447$, and $\theta_3=0.361$. The carrier frequencies
and bandwidths associated with these sources are set to
$f_1=152\text{MHz}$, $f_2=323\text{MHz}$, $f_3=432\text{MHz}$,
$B_1=20\text{MHz}$, $B_2=20\text{MHz}$, and $B_3=15\text{MHz}$.
The complex baseband signals are generated by passing the complex
white Gaussian noise through low-pass filters with different
cutoff frequencies. Also, the number of data samples used for
calculating the correlation matrices is set to $N_s=10^5$. The
sampling rate $f_s$ is chosen to be $f_s=28\text{MHz}$, which is
slightly higher than the minimum sampling rate $f_s\geq
B=\max\{B_1,B_2,B_3\}$ required for perfect recover of the power
spectrum of the wide frequency band. The SNR is set to 5dB. Fig.
\ref{fig:example1}(a) shows the true (marked with `$\square$') and
the estimated (marked with `$+$') carrier frequencies and DoAs for
the three sources. We can see that the estimated carrier
frequencies and DoAs coincide with the groundtruth well. Fig.
\ref{fig:example1}(b) and Fig. \ref{fig:example1}(c) respectively
depict the original power spectrum and the estimated power
spectrum of the wide frequency band. It can be observed that our
proposed method, even with a low SNR and a sampling rate far below
the Nyquist rate, is able to accurately identify the locations of
the occupied bands.

\begin{figure}[!t]
 \centering
\subfigure[True and estimated carrier frequencies and DoAs.]
{\includegraphics[width=3.5in]{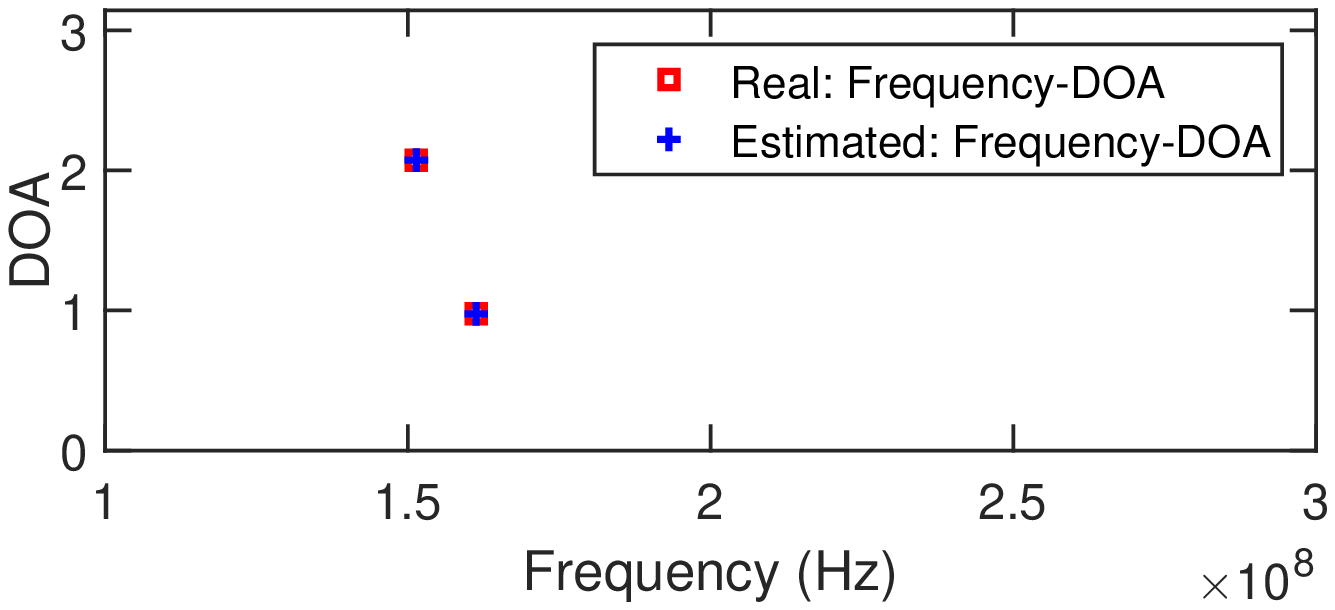}} \\
\subfigure[Original power spectra of sources.]
{\includegraphics[width=3.5in]{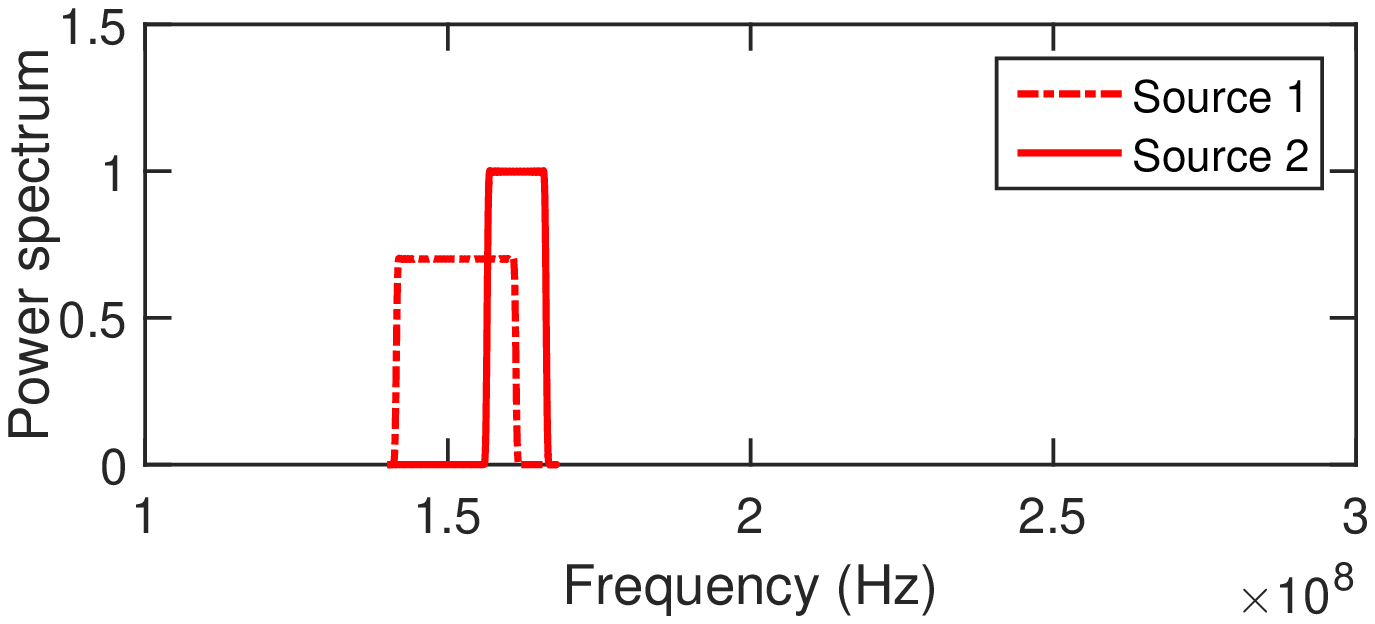}} \\
\subfigure[Estimated power spectra of sources.]
{\includegraphics[width=3.5in]{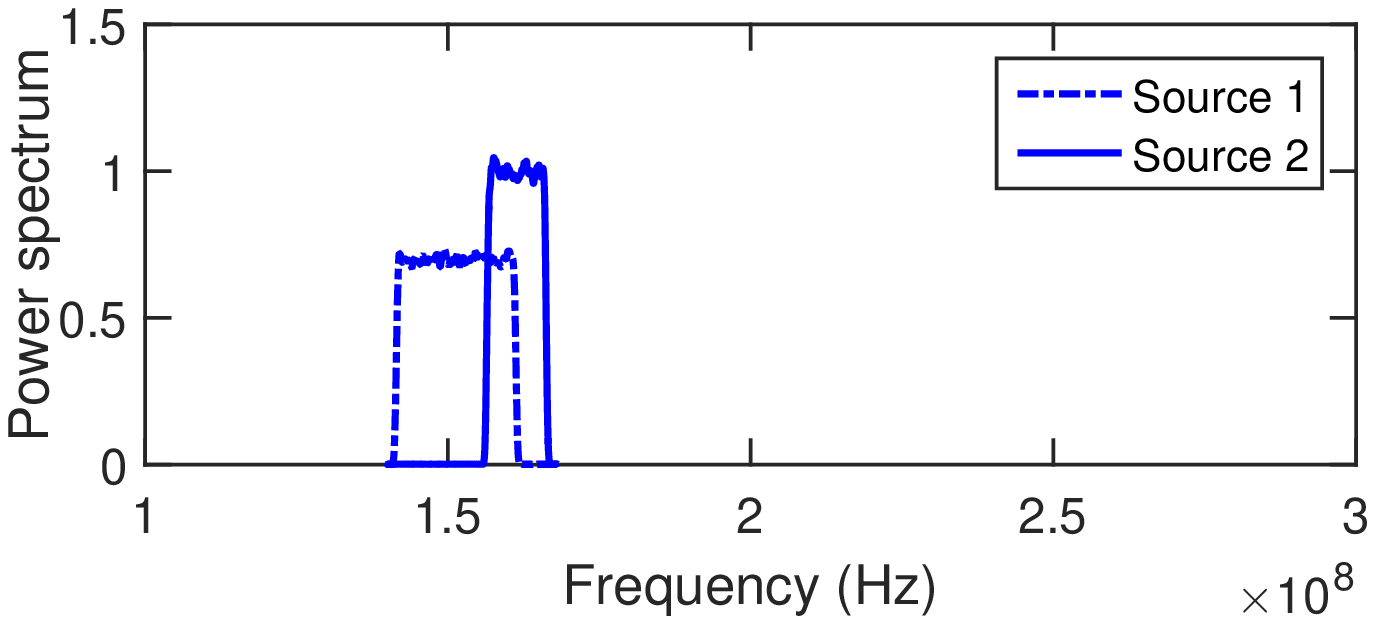}}
  \caption{Estimated carrier frequencies, DoAs and power spectra for sources that have
  partial spectral overlap, $\text{SNR}=20$dB.}
   \label{fig:example2}
\end{figure}

Next, we examine the scenario where frequency bands of the
narrowband sources overlap each other. Set $K=2$. The DoAs of
these two sources are given respectively by $\theta_1=2.064$ and
$\theta_2=0.968$. The carrier frequencies and bandwidths
associated with these two sources are set to $f_1=
151.36\text{MHz}$, $f_2=161.36\text{MHz}$, $B_1=20\text{MHz}$, and
$B_2=10\text{MHz}$. The power spectra associated with the two
sources are shown in Fig. \ref{fig:example2}(b), from which we can
see that the two sources partially overlap in the frequency
domain. The number of data samples $N_s$ and the sampling rate
$f_s$ remain the same as in the previous example. The SNR is set
to 20dB. The estimated carrier frequencies, DoAs, and the power
spectra of the two sources are plotted in Fig.
\ref{fig:example2}(a) and Fig. \ref{fig:example2}(c). We see that
our proposed method works well for sources with partially
overlapping frequency bands. This example shows that our proposed
method not only can perform wideband spectrum sensing, but also
has the ability to blindly separate power spectra of sources that
have partial spectral overlap.





\begin{figure}[!t]
 \centering
\subfigure[]
{\includegraphics[width=3.5in]{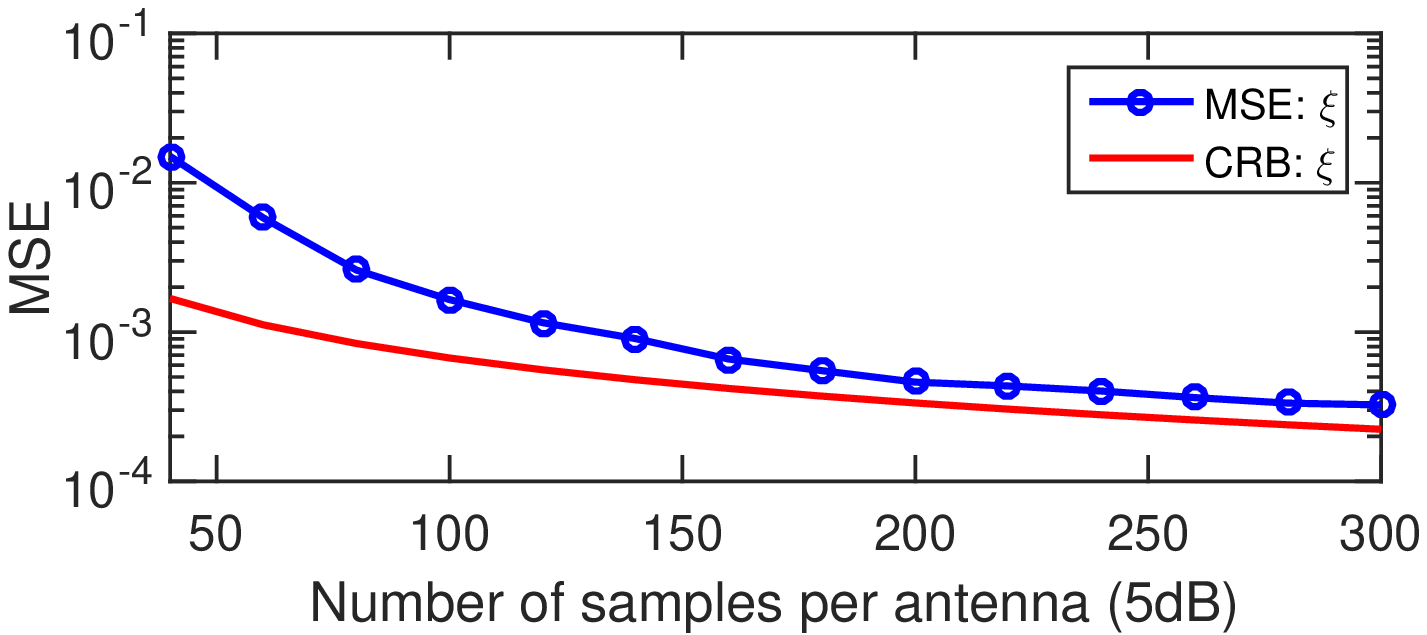}} \\
\subfigure[]
{\includegraphics[width=3.5in]{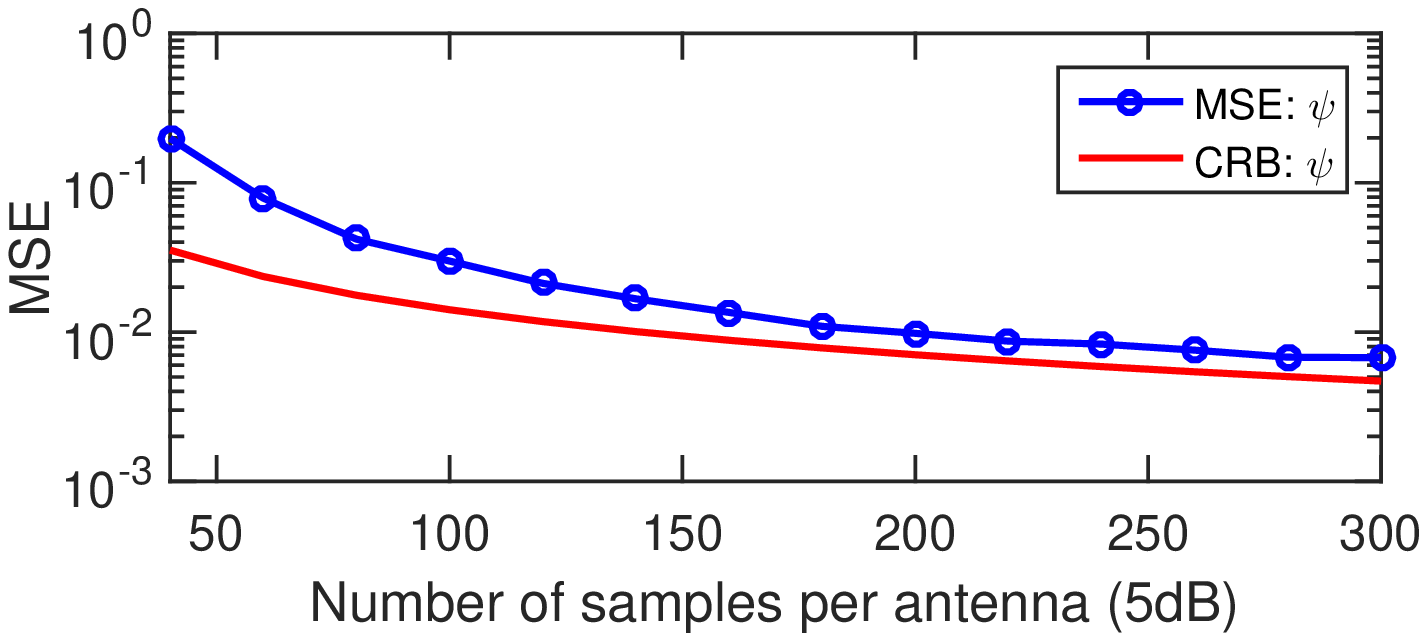}} \\
\subfigure[]
{\includegraphics[width=3.5in]{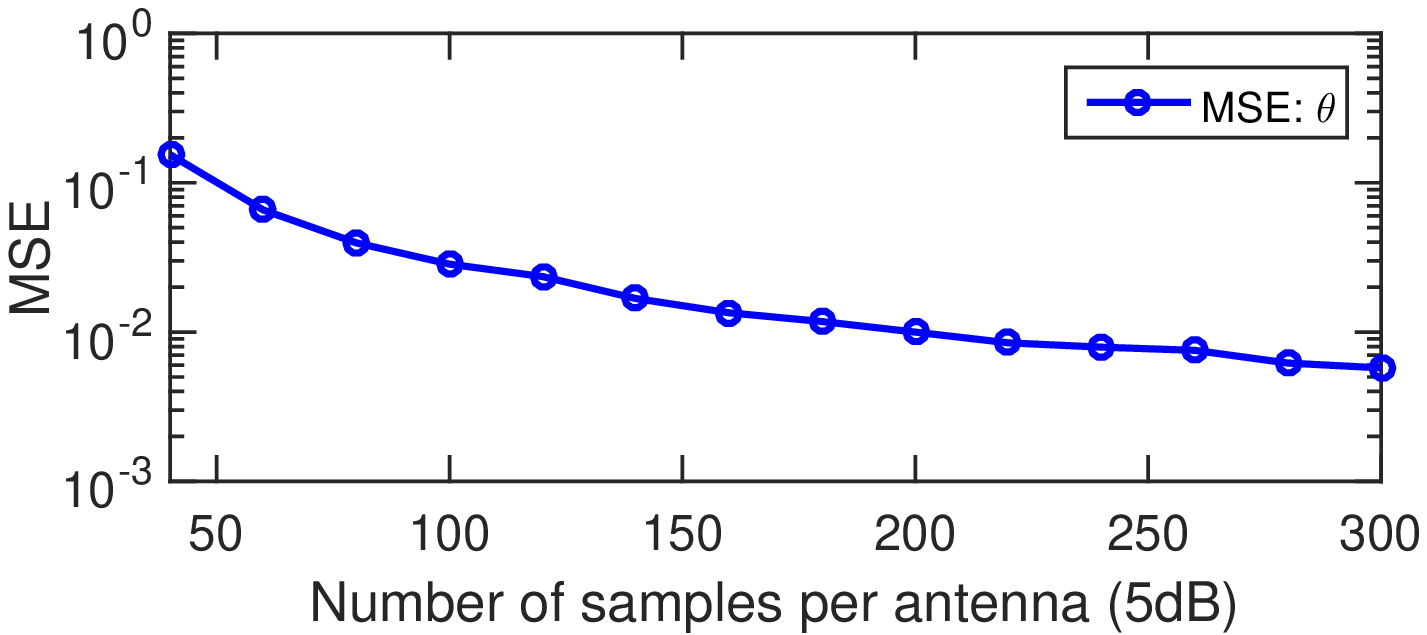}} \\
\subfigure[] {\includegraphics[width=3.5in]{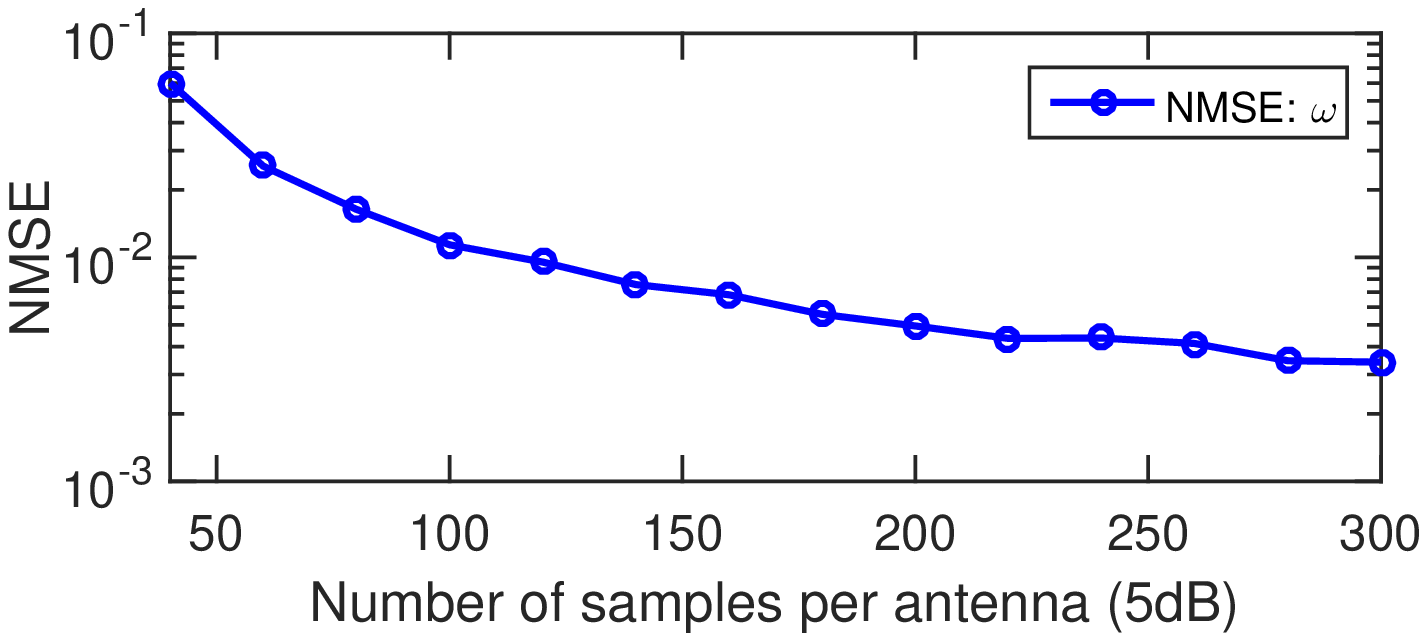}}
  \caption{MSEs and NMSE vs. the number of samples per antenna, where $N=8$ and $\text{SNR}=5$dB.}
   \label{fig:MSE-5dB}
\end{figure}

\begin{figure}[!t]
 \centering
\subfigure[]
{\includegraphics[width=3.5in]{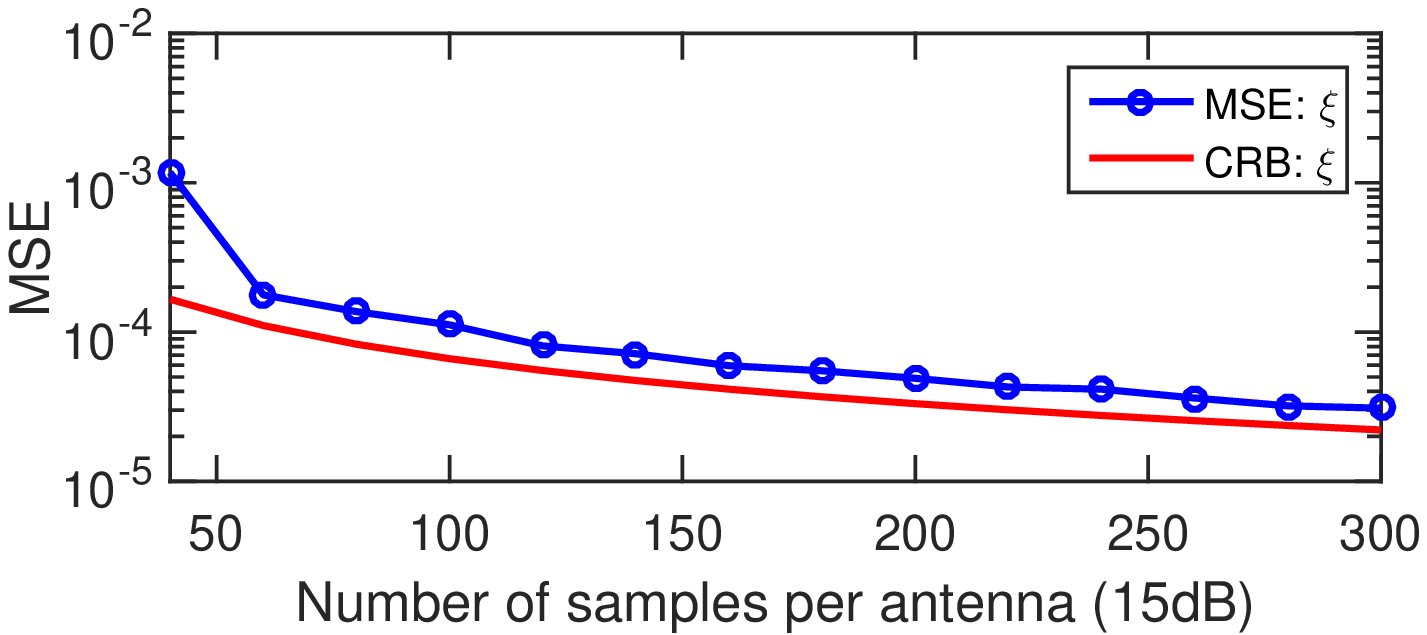}} \\
\subfigure[]
{\includegraphics[width=3.5in]{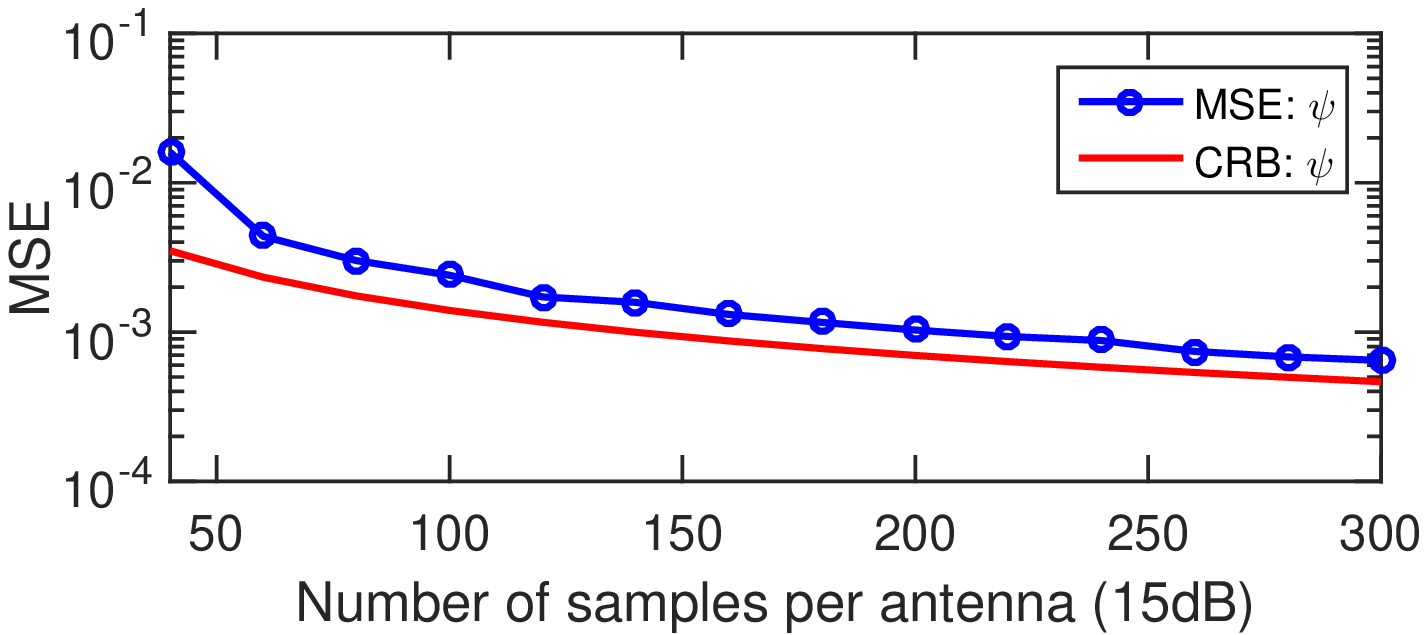}} \\
\subfigure[]
{\includegraphics[width=3.5in]{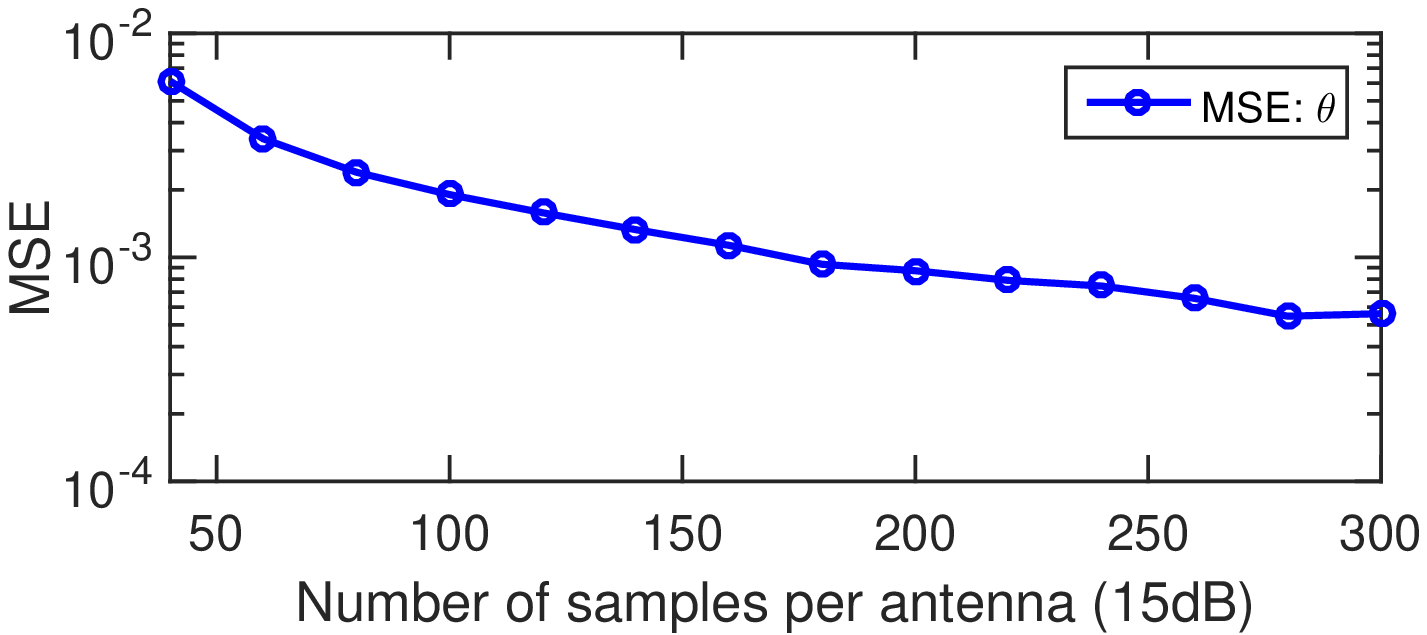}} \\
\subfigure[] {\includegraphics[width=3.5in]{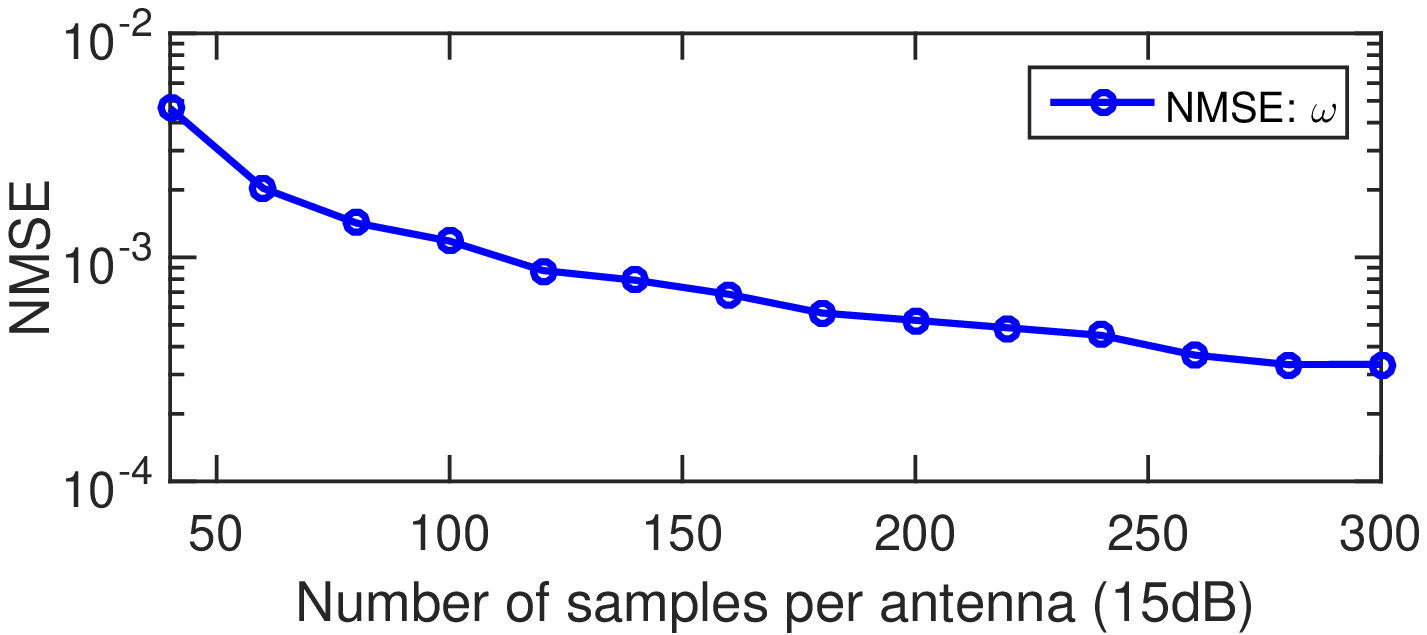}}
  \caption{MSEs and NMSE vs. the number of samples per antenna, where $N=8$ and $\text{SNR}=15$dB.}
   \label{fig:MSE-15dB}
\end{figure}

To better evaluate the performance of our proposed method, we
calculate the mean square errors (MSEs) for the following sets of
parameters
\begin{align}
\text{MSE}(\psi)&=\sum_{k=1}^{K} |\psi_k-\hat{\psi}_k|^2 \nonumber \\
\text{MSE}(\xi)&=\sum_{k=1}^{K} |\xi_k-\hat{\xi}_k|^2 \nonumber \\
\text{MSE}(\theta)&=\sum_{k=1}^{K} |\theta_k-\hat{\theta}_k|^2
\nonumber
\end{align}
Recalling that in our analysis, instead of concerning
$\{\theta_k,\omega_k\}$, we define two new parameters
$\xi_k\triangleq\omega_k\tau_k$ and $\psi_k\triangleq\omega_k/c$
and derive the CRB for $\{\xi_k,\psi_k\}$ in order to avoid the
numerical instability issue. The MSEs of the sets of parameters
$\{\xi_k,\psi_k\}$ are also included to compare with their
associated CRB results. The estimation accuracy of the carrier
frequencies is quantified by the normalized mean square error
(NMSE) defined as
\begin{align}
\text{NMSE}(\omega)=\sum_{k=1}^{K}
\frac{|\omega_k-\hat{\omega}_k|^2}{|\omega_k|^2} \nonumber
\end{align}
In this example, we set the number of sources $K=2$. The
parameters associated with these two sources are given as:
$f_1=152\text{MHz}$, $f_2=437\text{MHz}$, $B_1=126\text{KHz}$,
$B_2=63\text{KHz}$, $\theta_1=\pi/4$, and $\theta_2=\pi/3$. The
sampling rate is set to $f_s=1.26\text{MHz}$. Fig.
\ref{fig:MSE-5dB} and Fig. \ref{fig:MSE-15dB} depict the
MSEs/NMSEs of respective sets of parameters vs. the number of
samples $N_s$, where we set $\text{SNR}=5\text{dB}$ and
$\text{SNR}=15\text{dB}$, respectively. MSE/NMSE results are
averaged over 1000 independent runs, where the baseband complex
source signals are randomly generated for each run. We see that
our proposed method can achieve an estimation accuracy close to
the CRBs by using only a small number of data samples, e.g.
$N_s=200$. In Fig. \ref{fig:MSE-SNR}, we plot the MSEs/NMSEs of
different sets of parameters as a function of the SNR, where
$N_s=300$ data samples are used. We see that under a moderately
high SNR, say, $\text{SNR}=15\text{dB}$, our proposed method
attains an accurate estimate of the DoAs/carrier frequencies with
the MSE (NMSE) as low as $10^{-4}$.

\begin{figure}[!t]
 \centering
\subfigure[]
{\includegraphics[width=3.5in]{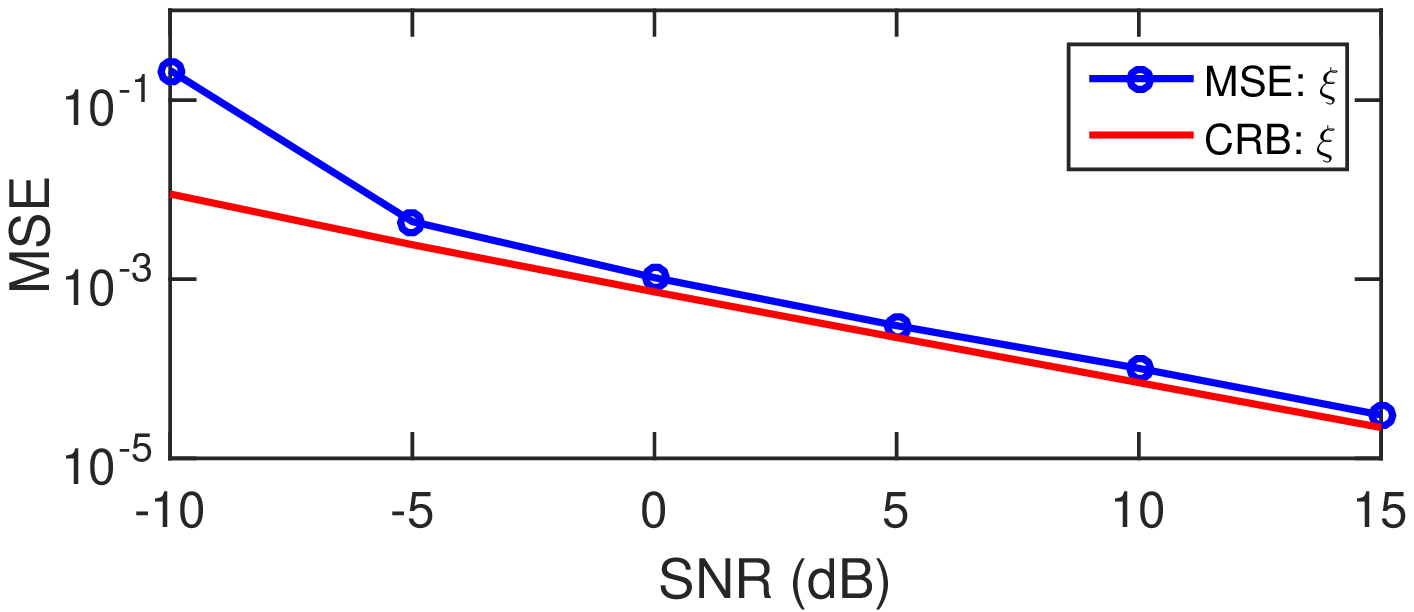}} \\
\subfigure[]
{\includegraphics[width=3.5in]{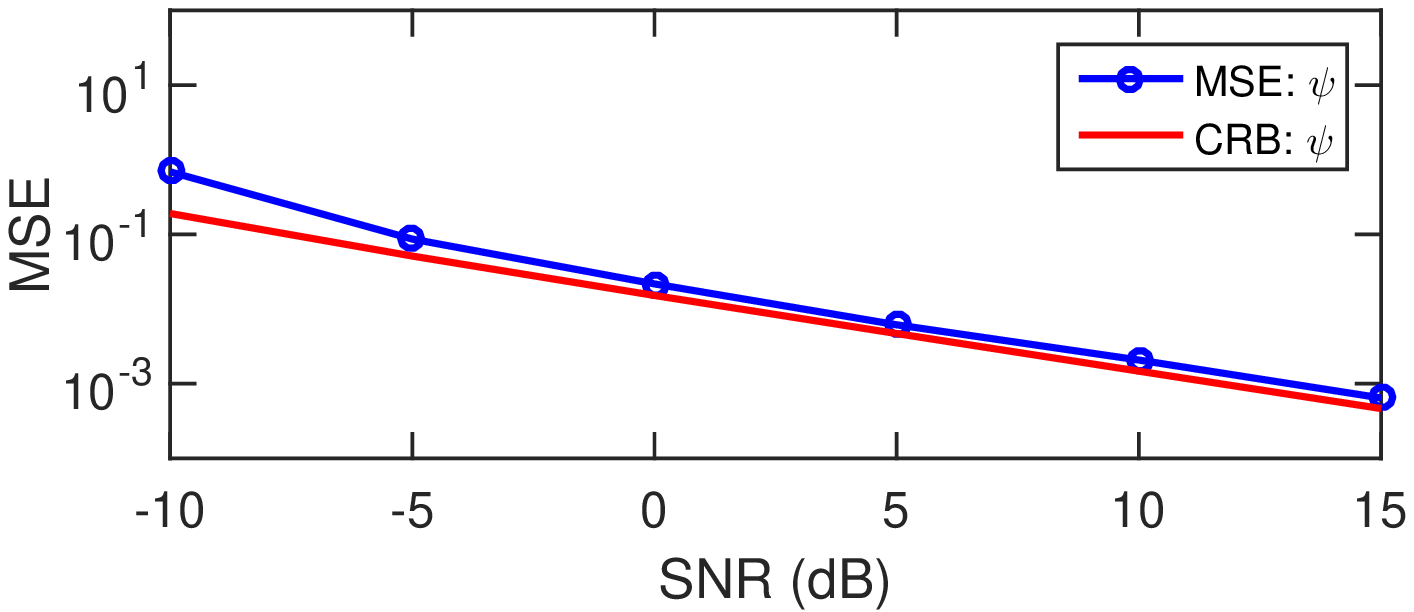}} \\
\subfigure[]
{\includegraphics[width=3.5in]{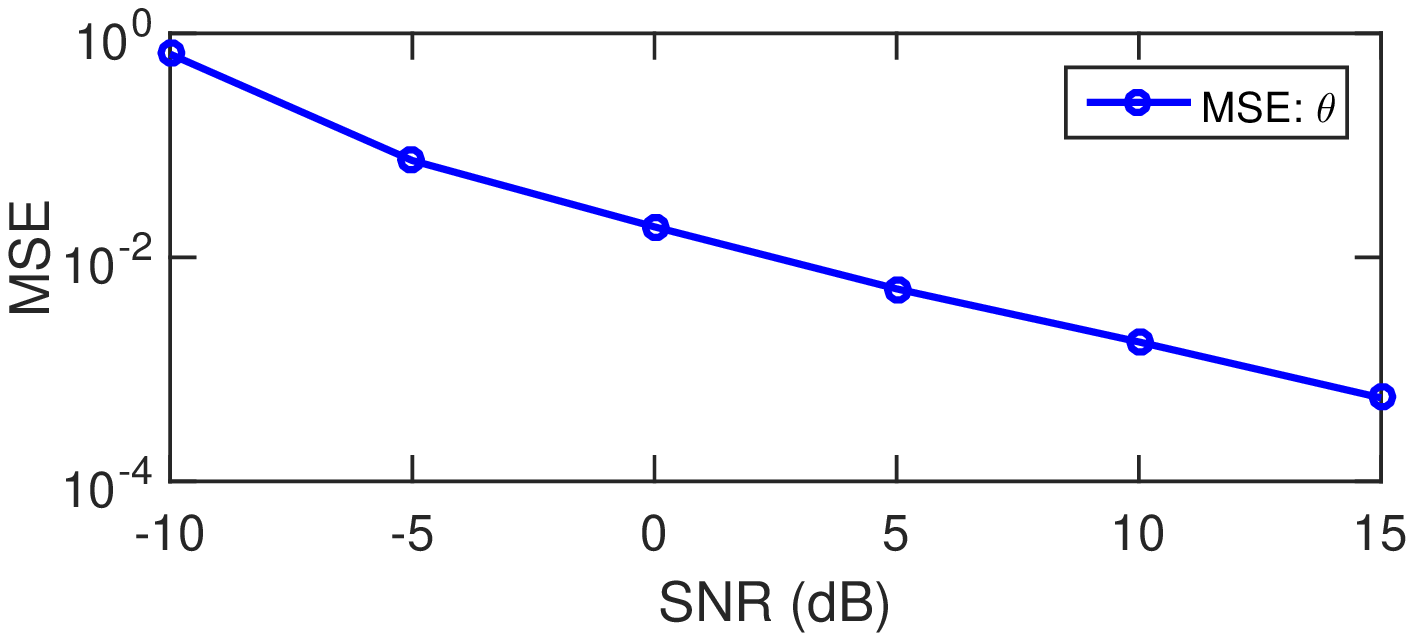}} \\
\subfigure[] {\includegraphics[width=3.5in]{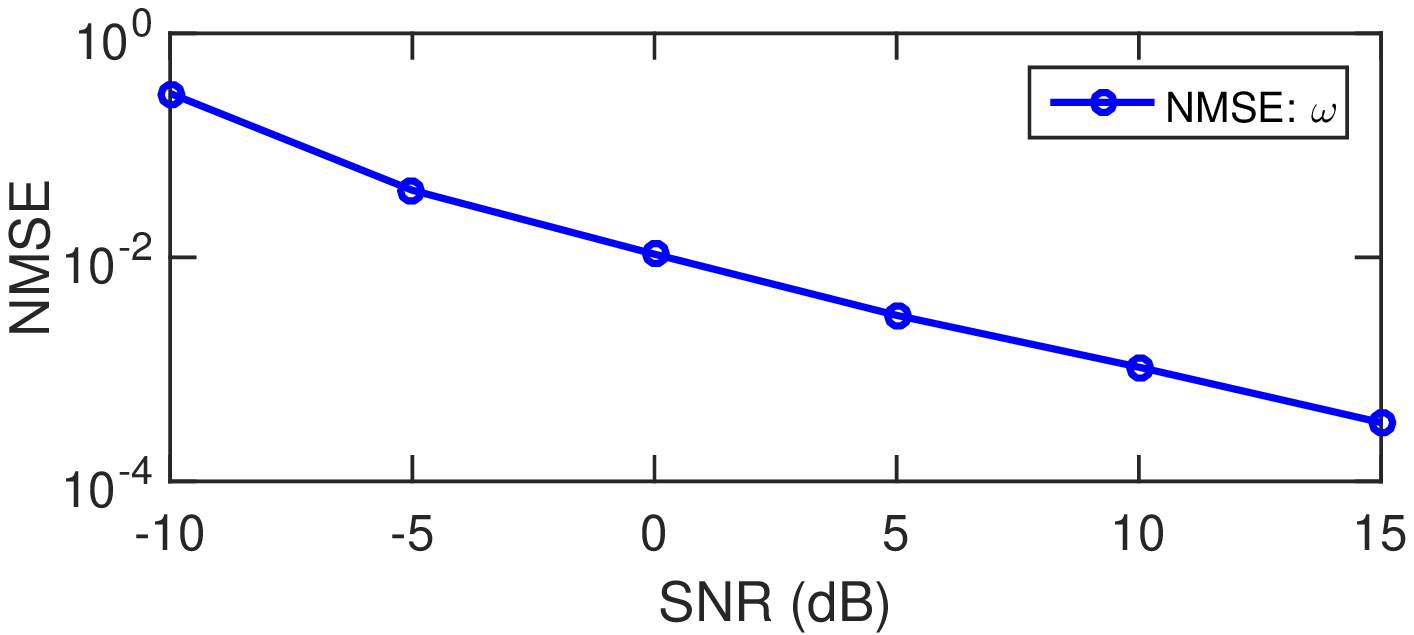}}
  \caption{MSEs and NMSE vs. SNR (dB), where $N=8$ and $N_s=300$.}
   \label{fig:MSE-SNR}
\end{figure}


\section{Conclusions} \label{sec:conclusion}
We considered the problem of joint wideband spectrum sensing and
DoA estimation in this paper. To overcome the sampling rate
bottleneck, we proposed a phased-array based sub-Nyquist sampling
architecture (termed as PASSAT) that is simpler in structure and
easier for implementation as compared with existing sub-Nyquist
receiver architectures. Based on the proposed receiver
architecture, we developed a CP decomposition-based method for
joint DoA, carrier frequency, and power spectrum estimation. The
conditions for exact recovery of the parameters and the power
spectrum were analyzed. Our analysis suggests that the perfect
recovery condition for our proposed method is mild: to recover the
power spectrum of the wide frequency band, we only need the
sampling rate to be greater than the bandwidth of the narrowband
source signal which has the largest bandwidth among all sources.
In addition, even for the case where sources have partial spectral
overlap, our proposed method is still able to extract the DoA,
carrier frequency, and the power spectrum associated with each
source signal. CRB analysis for our estimation problem was also
carried out. Simulation results show that our proposed method,
with only a small number of data samples, can achieve an
estimation accuracy close to the associated CRBs.

\bibliography{newbib}
\bibliographystyle{IEEEtran}

\end{document}